\documentclass[12pt,a4paper]{article}
\usepackage[latin1]{inputenc}
\usepackage[british]{babel}
\usepackage[left=80pt, right=80pt, top=60pt]{geometry}

\usepackage[affil-it]{authblk} %affiliation in title page
\usepackage{graphics}
\usepackage{amsfonts}
\usepackage{enumerate}
\usepackage{amssymb}
\usepackage{xfrac}	%small fractions
\usepackage{slashed} %\slashed{d}
\usepackage{bbm}	%\1
\usepackage{savesym}			%begin: overriding amsmath definition
\savesymbol{intertext}
\usepackage{amsmath}
\restoresymbol{TT}{intertext}	%end
\usepackage{upgreek}
\usepackage[bookmarksopen,bookmarksdepth=2]{hyperref}
\usepackage{threeparttable}
\usepackage{tipa}
\usepackage{color}

\usepackage{caption}
\usepackage{subcaption}
\usepackage{tikz}
\usepackage{pict2e}
\usepackage{wrapfig}

\mathchardef\ordinarycolon\mathcode`\:
\mathcode`\:=\string"8000
\begingroup \catcode`\:=\active
  \gdef:{\mathrel{\mathop\ordinarycolon}}
\endgroup

\newcommand{\N}{\mathbb N}
\newcommand{\Z}{\mathbb Z}
\newcommand{\Q}{Q_\hbar^\infty}

\newcommand{\Ql}{Q_\hbar^l}
\newcommand{\Qm}{Q_\hbar^m}
\newcommand{\Qn}{Q_\hbar^n}
\newcommand{\Qp}{Q_\hbar^p}
\newcommand{\R}{\mathbb R}
\newcommand{\C}{\mathbb C}
\newcommand{\T}{\mathbb T}
\newcommand{\A}{\mathcal{A}}

\renewcommand{\S}{\mathcal{S}}

\renewcommand{\H}{\mathcal{H}}
\renewcommand{\L}{\mathcal{L}}
\newcommand{\M}{\mathcal{M}}

\newcommand{\mR}{\mathcal{R}}
\newcommand{\Cr}{C_\mR}

\newcommand{\mB}{\mathcal{B}}
\newcommand{\W}{\mathcal{W}}
\newcommand{\Wr}{\mathcal{W}^0_\mR}

\newcommand{\supnorm}[1]{\norm{ #1 }_\infty}
\newcommand{\absnorm}[1]{\norm{ #1 }_1}
\newcommand{\quadnorm}[1]{\norm{#1}_2}

\newcommand{\norm}[1]{\left\| #1 \right\|}

\newcommand{\p}[2]{\langle #1 , #2 \rangle}

\newcommand{\spn}{\textnormal{span}}
\newcommand{\sgn}{\textnormal{sgn}}
\newcommand{\dom}{\textnormal{dom}}

\newcommand{\supp}{\operatorname{supp}}
\newcommand{\g}{\mathfrak{g}}
\renewcommand{\|}{\Vert} %I don't know why but the usual \| makes my latex crash
\newcommand{\add}{\textnormal{add}}
\newcommand{\sub}{\textnormal{sub}}
\newcommand{\conf}{\textnormal{conf}}
\newcommand{\mom}{\textnormal{mom}}
\newcommand{\vol}{\textnormal{vol}}
\newcommand{\id}{\textnormal{id}}
\newcommand{\indicator}{1}
\newcommand{\uvp}{\underline{\varphi}}
\usepackage{amsthm}
\newtheorem{thm}{Theorem}[section]
\newtheorem{lem}[thm]{Lemma}
\newtheorem{prop}[thm]{Proposition}

\newtheorem{defi}[thm]{Definition}

\newtheorem{exam}[thm]{Example}
\newtheorem{rema}[thm]{Remark}

\newcommand{\hatotimes}{\mathbin{\hat{\otimes}}}

\begin{document}

\title{Strict deformation quantization of abelian lattice gauge fields}
\author{Teun D.H. van Nuland
\\
\small Institute for Mathematics, Astrophysics and Particle Physics, Radboud University Nijmegen, Heyendaalseweg 135, 6525 AJ Nijmegen, The Netherlands. \\
E-mail: t.vannuland@math.ru.nl
}
\date{\today}
\maketitle

\begin{abstract}
	This paper shows how to construct classical and quantum field C*-algebras modeling a $U(1)^n$-gauge theory in any dimension using a novel approach to lattice gauge theory, while simultaneously constructing a strict deformation quantization between the respective field algebras. %This is believed to be a crucial step towards rigorous construction of quantum abelian gauge theories, and at the same time one of the few examples of a strict deformation quantization with infinite degrees of freedom. % We give the appropriate sense in which the continuum limit of abelian lattice gauge theory should be taken, when the lattice approaches the time-slice (e.g., $\R^3$) and thusly construct field algebras with infinite degrees of freedom.
	The construction starts with quantization maps defined on operator systems (instead of C*-algebras) associated to the lattices, in a way that quantization commutes with all lattice refinements, therefore giving rise to a quantization map on the continuum (meaning ultraviolet and infrared) limit. Although working with operator systems at the finite level, in the continuum limit we obtain genuine C*-algebras. We also prove that the C*-algebras (classical and quantum) are invariant under time evolutions related to the electric part of abelian Yang--Mills. Our classical and quantum systems at the finite level are essentially the ones of \cite{vNS20}, which admit completely general dynamics, and we briefly discuss ways to extend this powerful result to the continuum limit. We also briefly discuss reduction, and how the current set-up should be generalized to the non-abelian case.
\end{abstract}

\section{Introduction}

C*-algebras are expected to provide the building blocks of a mathematical construction of gauge theories such as QED and quantum Yang--Mills. Besides having already proven their worth in putting quantum mechanics on a solid basis, C*-algebras feature in the Haag--Kastler axioms \cite{Haag}, and could therefore be used to non-perturbatively construct a quantum field theory. Moreover, as C*-algebras can model both quantum and classical theories, a C*-algebraic model of a classical gauge theory might provide a good footing from which to take the leap towards a quantum gauge theory.
% Moreover, C*-algebras model both quantum and classical theories, notable because classical gauge theories are fully understood and form the starting point of a quantum formulation. 
The direction of this leap, then, might be indicated by strict deformation quantization \cite{Landsman98}, for it gives a set of axioms that a quantization map between a classical and a quantum C*-algebra should satisfy. These axioms are stringent, and examples are mostly found in finite-dimensional spaces \cite{Landsman93,Landsman98,LMvdV,Rieffel89,Rieffel}, with a few exceptions that usually rely on finite-dimensional approximations \cite{BHR,vN19}. To quantize a gauge theory, one is therefore advised to first quantize a finite-dimensional regularization, and this is where lattice gauge theory comes in.

Lattice gauge theory was introduced by Wilson in \cite{Wilson} and shows how to approximate gauge fields by their parallel transports on a lattice (where by `lattice' we mean a type of finite graph). Wilson's framework has had a huge impact, both in theoretical and phenomenological physics. On the theoretical side, an important contribution was made in \cite{KS} by Kogut and Susskind, who took a Hamiltonian approach to Wilson's ideas, considering lattices in a time-slice -- typically $\R^3$ -- and showed that the parallel transports of a gauge field on the lattice can be interpreted as rigid rotors, and that Yang--Mills time evolution implies a certain coupled movement of these rotors. %The Hamiltonian approach is very appealing to the operator algebras community because a finite dimensional Hamiltonian system is well understood there. 
Important for us, the finite-dimensionality of this quantum Hamiltonian system makes it suitable for the C*-algebraic approach. This C*-algebraic approach to Hamiltonian quantum lattice gauge theory is pursued for instance in \cite{ASS,BS19a,BS19b,GR,vNS20,Stienstra,ST}. A central goal of this program is to describe the continuum limit (in which the lattices are replaced by the full $\R^3$ or a subset thereof) by a C*-algebra invariant under a *-homomorphism coming from Yang--Mills dynamics. Such a continuum C*-algebra has the potential to give rise to a local quantum field theory.

The current paper will add to this program by constructing promising new field algebras for quantum abelian lattice gauge theories in arbitrary dimension, using an approach guided by C*-algebraic quantization.
%(For the reader living by the mantra that `the world is quantum'. These readers can skip Section 3.)
 %
 It follows up on \cite{vNS20}, where the field algebra corresponding to a quantum abelian gauge theory on a fixed lattice is defined as the closure of the image under a Weyl quantization map of a classical algebra that is the analogue of the commutative resolvent algebra \cite{BG,vN19} when replacing the configuration space $\R^{nk}$ by $\T^{nk}$. Here $\T^n$ is interpreted as the abelian gauge group and $k$ as the number of edges. The obtained field algebra, named the `resolvent algebra on the cylinder' in \cite{vNS20} is a C*-algebra of bounded operators on $L^2(\T^{nk})$, naturally containing a copy of the crossed product algebra $C(\T^{nk})\rtimes \T^{nk}$ as a C*-subalgebra. The main advantage of the resolvent algebra on the cylinder is that it is conserved under a very general class of time evolutions \cite[Theorem 29]{vNS20}. Independently of what C*-algebra one takes at the finite level, there appeared several problems, on the side of quantum embedding maps as well as on the side of quantization. The embedding maps for adding an edge to the lattice are easily defined by construction of the algebras in \cite{vNS20}, but for subdivision of edges we could not find a natural embedding map, for reasons explained in \cite[pages 247--249]{Stienstra}. Moreover, the quantization map was not a strict deformation quantization, lacking injectivity as well as Rieffel's condition. 
 The current paper solves all of the above problems simultaneously, by letting go of multiplicativity of the embedding maps.
 
  On each lattice, we restrict ourselves to a subspace of the classical C*-algebra on which the quantization map of \cite{vNS20} is injective. This subspace and its image under quantization turn out to be only operator systems, and not algebras. At first, this appears to distance us from the powerful C*-algebraic approach. However, on these operator systems, both the classical and quantum embedding maps are now naturally defined and commute with quantization. Moreover, the ensuing limit of operator systems turns out be a *-algebra lying dense in a C*-algebra, thus recovering the C*-algebraic approach.   
  
  This `operator systemic' method has many advantages. The obtained quantum embedding maps respect the gauge action, which becomes very important when one wishes to make the step from field algebras to observable algebras.
  Moreover, in the continuum limit, the quantum and classical theory behave even better than in the case on the lattice, in the sense that they form a strict deformation quantization, satisfying all conditions of \cite[Definition II.1.1.1 and II.1.1.2]{Landsman98}.

For these reasons, the operator systemic method seems to improve upon the existing literature. In most operator algebraic approaches to lattice gauge theory (e.g., \cite{ASS,BS19a,BS19b,GR,Stienstra,ST}) one uses inductive limits of C*-algebras instead. We validate our deviation in \textsection 2.2.

The emergence of a strict deformation quantization counts as another validation of our method, but is also remarkable in itself. Most notably, it involves two limits; besides the usual limit $\hbar\to0$ also the limit of lattice spacing tending to zero becomes important. The interaction between these two limits complicates the proof at most places, but in other places is the very reason the result holds.

Section \ref{sct:2} of this paper constructs the classical C*-algebra on the continuum, the quantization map on the continuum, and the quantum C*-algebra on the continuum. The classical and quantum C*-algebras are shown to be invariant under time evolution related to the electric part of abelian Yang--Mills \cite{KS} in \textsection 2.4. Section \ref{sct:SDQ} gives the proof of strict deformation quantization, and forms by far the most technical part of this paper. Section \ref{sct:outlook} provides a positive outlook on three logical next steps, namely reduction, full time evolution, and generalization to non-abelian gauge groups.

%The use of operator systems to obtain a continuum C*-algebra seems to be a crucial and novel step towards quantization of gauge theories. In most operator algebraic approaches to lattice gauge theory, one uses direct (i.e. inductive) limits of C*-algebras instead. This is seen for instance in \cite{ASS,BS19a,BS19b,GR,Stienstra,ST}. We argue that operator systems are superiour when it comes to defining quantum embedding maps in Section \ref{sct:quantum systems and embedding maps}.

%There are two logical next steps to take in order to strengthen the results of this paper, that seem to be more accesible than full Yang--Mills time evolution. 

%Further research is expected to improve the results of the current paper in two ways.
%Firstly, the field algebras need to be replaced by observable algebras obtained from the respective field algebras by reduction. Fortunately, gauge action is respected by the quantum embedding maps. Secondly, the whole framework needs to be generalized to cover nonabelian gauge groups like $SU(3)$.
%%
%At the end of this paper we give an encouraging outlook on these two steps.
\vspace{-10pt}
\paragraph{Notation} We denote $G:=\T^n:=\R^n/\Z^n$, $\g:=\R^n$ and $\g^*:=\R^n$. Elements of $G^l$ for a set $l$ are usually denoted by $q$ or $[x]$ where $[x]:=x+(\Z^n)^l$ for $x\in(\R^n)^l$. We denote by $L_q$ the left-translation on $G^l$, i.e., $L_{[x]}[y]=[x+y]$. We denote by $M_g$ the multiplication operator of the function $g$. We denote by $e^{i\xi\cdot}$ the function $x\mapsto e^{i\xi\cdot x}$ and, slightly abusing notation, by $e^{2\pi i a\cdot}$ the function $[x]\mapsto e^{2\pi i a\cdot x}$ for $a\in(\Z^n)^l$. We denote by $\psi_a$ the equivalence class of $e^{2\pi ia\cdot}$ in $L^2(G^l)$, the Hilbert space of square-integrable functions, and by $\mB(\H)$ the bounded linear operators on any Hilbert space $\H$. We denote by $C(X),C_b(X),C^\infty(X),C_c^\infty(X),\S(X)$ respectively the continuous functions, the bounded ones, the smooth functions, the compactly supported ones, and the Schwartz functions on $X$. In any metric space, $B_d(x)$ is the open ball around $x$ with radius $d$. We let $B:=B_{1/2}(0_\g)\subseteq\g$, remarking that $x\mapsto [x]$ is a diffeomorphism on $B$. By an operator system we mean a linear subspace of a unital C*-algebra that is preserved under $*$ and contains $1$. We do not require operator systems to be closed.

\vspace{-10pt}
\paragraph{Acknowledgements}
I am grateful to Walter van Suijlekom for providing indispensable constructive feedback and to Klaas Landsman for providing indispensable enthusiasm. Research supported by NWO Physics Projectruimte (680-91-101).

\section{Operator systems and limit C*-algebras}
\label{sct:2}

\paragraph{Lattices}
Let us first define what we mean by `a lattice'. %Our notation is closest to \cite{ASS,Stienstra}, but can also be related to \cite{GR,ST}. 
%As a time-slice, which is often chosen to be $\R^3$, we use a general metric space $\mathcal{M}$. 
For simplicity, we take our time-slice to be $\R^D$ in this paper, although any metric space would work.
Throughout this article, a \emph{lattice} is a finite subset $l\subseteq\R^D\times\R^D$ such that, using the lexicographical ordering of $\R^D$, we have $x<y$ for all $(x,y)\in l$, and, we have $tx+(1-t)y\neq sz+(1-s)w$ for all $(x,y),(z,w)\in l$ and all $0<t,s<1$. The elements $e=(x,y)$ of a lattice $l$ are interpreted as directed straight edges from $x$ to $y$. Thus, all we ask of a lattice is that its edges do not intersect, except possibly at their boundaries. The set of all lattices becomes a directed set, denoted $(\L,\leq)$, when we agree that $l\leq m$ if and only if the lattice $m$ can be obtained from $l$ by adding and subdividing edges in the sense of \cite{ASS}. Put precisely, $l\leq m$ if and only if for all $(x_1,x_2)\in l$ there exists $N\in\N_0$ and $0<t_1<\cdots<t_N<1$ such that for $y_s:=(1-t_s)x_1+t_sx_2$ we have $(x_1,y_1),(y_1,y_2),\ldots,(y_{N-1},y_N),(y_N,x_2)\in m$.
We endow every edge $e=(x,y)$ with a length $d_e:=\norm{x-y}$.

%
%\begin{rema}
%	In particular, $\L$ defines the objects of a category, where the morphisms (from $l\in \L$ to $m\in \L$, say) are all pairs $(l,m)$ with $l\leq m$. It is concievable that one could formalize and generalize this whole paper using categorical language.
%\end{rema}

 Let us compare our notation with the one in \cite{ASS,Stienstra}, in which an index set $I$ is used, and $\{\Lambda_i\}_{i\in I}$ is the net of finite lattices, including a set of vertices $\Lambda_i^0$, a set of edges $\Lambda_i^0$, and a set of plaquettes $\Lambda_i^2$. In our situation, the elements $l\in\L$ can be identified with the sets of edges $\Lambda^1_i$. As in this paper we will not reduce to the gauge group and only discuss the electric part of Yang--Mills dynamics, the vertices and plaquettes will play no role. By our definition of $l\in\L$ and simply following set notation, $G^l$ denotes the set of functions from the edges in $l$ to elements in $G$, or equivalently ordered tuples of length $|l|$ with elements in $G$.

%\subsection{Classical field algebra and operator system on a finite lattice}
%To each (finite) lattice $l\in\L$, we associate a commutative *-algebra
%\begin{align*}
%	\A^l_0:=\{\sum_j g_j\otimes e^{i\xi_j\cdot}:~g_j\in C^\infty(G^l),\xi_j\in \g^l\}
%\end{align*}
%which is densely included in the C*-algebra
%\begin{align*}
%	A_0=\overline{\A_0}=C(G^l)\hatotimes\W^0((\g^*)^l)
%\end{align*}
%where the notation $\W^0(E)=\W(E,0)$ is adapted from \cite{BHR}.
%
%We identify a *-preserved linear subspace $\M^l_0\subseteq\A^l_0$ by denoting $B:=B_{1/2}(0)\subseteq\g$ for the ball with radius 1/2 around $0\in\g$, and defining
%\begin{align*}
%	\M_0^l:=\spn\{ g\otimes e^{i\xi\cdot}:~g\in C^\infty(G^l),~\xi\in B^l\} \subseteq C^\infty_b(X^l).
%\end{align*}
%
%\subsection{The quantization map on a (finite) lattice}
%Let $l\in\L$ be a lattice. The quantization map $\Ql:\M_0^l\to\M_\hbar^l$ is defined by
%
%\begin{align*}
%		\Ql(e_k\otimes h)\psi_a=h(2\pi\hbar(a+\tfrac{1}{2}k))\psi_{a+k}.
%	\end{align*}

\subsection{The finite and continuum classical systems}

\paragraph{The continuum phase space}
	 Throughout this paper, we let $G:=\T^n$ as a Lie group. This is the configuration space associated to each edge of a lattice. The (abelian) Lie algebra of $G$ is $\g=\R^n$, and the exponential map $\g\to G$ is denoted $x\mapsto [x]$. The phase space $X^l$ associated to a lattice $l\in\L$ is given by the cotangent bundle of the Lie group $G^l$, i.e., $X^l:=T^*G^l\cong G^l\times(\g^*)^l$. In order to define connecting maps between $X^l$ and $X^m$, for lattices $l\leq m\in\L$, we use the fact that $m$ can be obtained from $l$ by recursively applying two operations: adding an edge to the lattice and subdividing an edge of length $d$ into two edges of lengths $d_1$ and $d_2$ with $d_1+d_2=d$. In that manner, we define connecting maps
\begin{align*}
%	\gamma_{lm}^{\conf}&:G^m\to G^l\\
%	\gamma_{lm}^{\mom}&:(\g^*)^m\to(\g^*)^l
	\gamma_{lm}=(\gamma_{lm}^\conf,\gamma_{lm}^\mom):G^m\times(\g^*)^m\to G^l\times(\g^*)^l
\end{align*}
by recursively composing embedded versions of the maps $\gamma_{\text{add}}=(\gamma^{\conf}_{\text{add}},\gamma^{\mom}_{\text{add}}):G^2\times(\g^*)^2\to G\times\g^*$ and $\gamma_{\text{sub}}=(\gamma^{\conf}_{\text{sub}},\gamma^{\mom}_{\text{sub}}):G^2\times(\g^*)^2\to G\times\g^*$ defined by
%\begin{wrapfigure}{l}{0.1\textwidth}
%	\begin{center}
%		%\includegraphics[scale=0.01]{bovencirkel.pdf}
%		%\begin{figure}
%		\begin{picture}(1,1)
%		\thicklines
%			\put(0,-6){\line(0,1){14}}
%			\put(0,-6){\circle*{3}}
%			\put(0,8){\circle*{3}}
%			\put(0,1){\circle*{3}}
%			
%			\put(20,-6){\line(0,1){7}}
%			\put(20,-6){\circle*{3}}
%			\put(20,1){\circle*{3}}
%			
%			\put(0,-32){\line(0,1){14}}
%			\put(0,-32){\circle*{3}}
%			\put(0,-18){\circle*{3}}
%			\put(0,-23){\circle*{3}}
%			
%			\put(20,-32){\line(0,1){14}}
%			\put(20,-32){\circle*{3}}
%			\put(20,-18){\circle*{3}}
%		\end{picture}
%	\end{center}
%\end{wrapfigure}
\begin{align*}
	\gamma^\conf_\add([x_1],[x_2])&:=[x_1];&\gamma^\mom_\add(v_1,v_2)&:=v_1;\\
	\gamma^\conf_\sub([x_1],[x_2])&:=[x_1+x_2];&\gamma^\mom_\sub(v_1,v_2)&:=\frac{d_1v_1+d_2v_2}{d}.
\end{align*}
These embedding maps arise naturally by interpreting $x_e\in G$ as the parallel transport along the edge $e$ and $v_e\in\g$ as the average rate of change along $e$. One could replace `average' by `total', at the cost of a slightly different quantization map. By construction, the maps $\gamma_{lm}:X^m\to X^l$ for $l\leq m\in\L$ define an inverse system of phase spaces.
%\begin{align*}
%	X^l:=T^*G^l\cong G^l\times(\g^*)^l,\quad
%	%G^{\infty}:=&\lim_{\rightarrow} G^l,\\
%	\gamma_{lm}=(\gamma_{lm}^{conf},\gamma_{lm}^{mom}):X^m\to X^l.
%\end{align*}
The ensuing inverse limit is denoted as
\begin{align*}
	X^\infty:=\lim_{\leftarrow} X^l=\lim_{\leftarrow}G^l\times\lim_{\leftarrow}(\g^*)^l,\quad
	\gamma_l=(\gamma_l^{conf},\gamma_l^{mom})
	:X^\infty\to X^l.
\end{align*}

\paragraph{Operator systems}
The classical system on the lattice $l\in\L$ can be described by the commutative C*-algebra introduced in \cite{vNS20}, namely
\begin{align*}
	A^l_0:=\Cr(T^*G^l)=C(G^l)\hatotimes \Wr((\g^*)^l),
\end{align*}
where $\Wr((\g^*)^l)$ is the C*-subalgebra of $C_b((\g^*)^l)$ generated by the commutative Weyl C*-algebra $\W((\g^*)^l,0)$ from \cite{BHR} and the commutative resolvent algebra $\Cr((\g^*)^l)$ from \cite{vN19}. The reason to work with the unital C*-algebra $A^l_0$ is that $A_0^l$ and its Weyl quantization are conserved under fully general dynamics in the sense of \cite{vNS20}. In contrast, the C*-subalgebra $C(G^l)\hatotimes\W((\g^*)^l,0)\subseteq A^l_0$, where $\W((\g^*)^l,0):=\overline{\spn}\{e^{i\xi\cdot}:~\xi\in\g^l\}$, is only conserved under `free' time evolution \cite{vNS20}. As explained in \cite{vNS20}, $A_0^l$ is the closure of the *-algebra $\A_0^l$ defined by
\begin{align*}
	\A_0^l:=\spn\left\{e^{2\pi i b\cdot}\otimes e^{i\xi\cdot}(g\circ P_V):~\begin{aligned}&b\in(\Z^n)^l,~ V\subseteq \g^l\text{ linear},\\ &g\in \S(V),~\hat{g}\in C_c^\infty(V^*),~\xi\in \g^l\end{aligned}\right\}.
\end{align*}
For this paper, we will only need that any element of $\A_0^l$ can be written as $\sum_{k=1}^K g_k\otimes h_k$ with $h_k\in C_b((\g^*)^l)$ a Fourier transform $h_k=\hat{\mu}_k:=\int d\mu_k(\xi)e^{i\xi\cdot}$ of a compactly supported finite complex Borel measure $\mu_k$ on $\g^l$. We can thusly define the operator system
\begin{align*}
	\M_0^l:=\spn\{g\otimes \hat\mu\in\A_0^l:~\supp(\mu)\subseteq B^l\}\subseteq\A_0^l,
\end{align*}
where $B=B_{1/2}(0_\g)$.
%Define the operator system
%\begin{align*}
%	\M_0^l:=\spn\{ e_a\otimes h:~a\in(\Z^n)^l,h\in\tilde\Wr((\R^n)^l) \} \subseteq C^\infty_b(X^l)
%\end{align*}
%where $e_a[x]:=e^{2\pi i a\cdot x}$ and
%\begin{align*}
%	\tilde\Wr((\R^n)^l):=\spn\bigg\{e^{i\xi\cdot}(g\circ P_V):~
%	\begin{array}{l}
%	V\subseteq (\R^n)^l\text{ linear, } g\in \S(V),\\
%	\supp(\hat{g})\subseteq B_{1/2}(0),~
%	\xi\in B_{1/2}(0)\subseteq V^\perp
%	\end{array}
%	\bigg\}.
%\end{align*}
%Functions $h\in\tilde\Wr((\R^n)^l)$ can be written as Fourier transforms of finite complex measures $\mu$ supported in $B_{1/2}(0)$.
%where $e^{i\xi\cdot}\in C((\g^l)^*)$ is defined as $\theta\mapsto e^{i\theta(\xi)}$
%The operator system $\M_0^l$ generates a *-algebra of the form
%\begin{align*}
%	\A_0^l:=\spn\{ g\otimes h:~g\in C^\infty(G^l),~h\in\Wr((\R^n)^l)\},
%\end{align*}
%which in turn lies dense in the C*-algebra $\Cr(T^*(\T^n)^l)$ defined and investigated in \cite{vNS20}.
The *-algebras $\A^l_0$ are endowed with the connecting maps $\gamma_{lm}^*:\A_0^l\to\A_0^m$, whose restrictions to the operator systems $\M_0^l$ we denote as 
	$$F^{ml}_C:=\gamma_{lm}^*|_{\M^l_0}:\M_0^l\to\M_0^m,$$
and refer to as the classical embedding maps.
We define the *-algebraic direct limit 
	$$\A_0^\infty:=\lim_{\rightarrow}\A_0^l,$$
and identify $\A_0^\infty\subseteq C_b(X^\infty)$ by identifying $F_C^l:\A_0^l\to\A_0^\infty$ with the restriction of $\gamma_l^*:C_b(X^l)\to C_b(X^\infty)$. To describe $\A_0^\infty$, it turns out we only need to regard the operator systems $\M_0^l$. To prove this, we first introduce the following useful notation.

\begin{defi}\label{def:l_R}
	For a lattice $l$ and a positive integer $R$, we let $l^R\geq l$ be the lattice obtained by subdividing every edge of $l$ into $R$ edges of equal length. 
\end{defi}

\begin{lem}
The direct limit of *-algebras $\A_0^l$ is also the direct limit of the operator systems $\M_0^l$, in the sense that we have
\begin{align}\label{eq:classical algebra vs operator system}
	\A_0^\infty=\{f\circ\gamma_l:~l\in \L,~f\in\M_0^l\}.
\end{align}
\end{lem}
\begin{proof}
By recursively composing the maps
\begin{align}\label{S^ml}
	S^\sub(\xi):=\left(\frac{d_1}{d}\xi,\frac{d_2}{d}\xi\right),\qquad S^\add(\xi):=(\xi,0),
\end{align}
we obtain a direct system of linear maps $S^{ml}:\g^l\to\g^m$ ($l\leq m\in\L$) allowing us to write the classical embedding maps as
\begin{align}\label{eq:F_C^ml in termen van S^ml}
	F_C^{ml}(g\otimes \hat\mu)=(g\circ\gamma^\conf_{lm})\otimes \widehat{S^{ml}_*\mu}.
\end{align}
For every $F_C^l(f)=f\circ\gamma_l\in\A_0^\infty$ we can write $f=\sum_k g_k\otimes \hat{\mu}_k$ for compactly supported measures $\mu_k$. Choose $R$ such that $\supp(\mu_k)\subseteq B_{R/2}(0_\g)^l$ for all $k$, and consider the lattice $l^R\geq l$. Every edge $e\in l^R$ satisfies $d_e=\tfrac{1}{R} d_{e'}$ for the edge $e'\in l$ it lies in. Hence 
	$$(S^{l^Rl}\xi)_e=\tfrac{1}{R}\xi_{e'}\,,$$
and so $S^{l^Rl}(\supp(\mu_k))\subseteq S^{l^Rl}(B_{R/2}(0_\g)^l)\subseteq B_{1/2}(0_\g)^{l^R}=B^{l^R}$. As $S^{l^Rl}$ is a closed map, we therefore obtain $\supp(S^{l^Rl}_*\mu_k)\subseteq B^{l^R}$ for all $k$. Then \eqref{eq:F_C^ml in termen van S^ml} gives $f\circ\gamma_{ll^R}\in\M^{l^R}_0$, so $F_C^l(f)=(f\circ\gamma_{ll^R})\circ\gamma_{l^R}$ is in the set on the right hand side of \eqref{eq:classical algebra vs operator system}.
\end{proof}

%\subsection{Poisson structure}\label{sct:Poisson structure}
\begin{rema}\label{rema:supremum trick}
	Two arbitrary functions in $\A_0^\infty$ can be written as $f_1\circ\gamma_l,f_2\circ\gamma_l$ for a certain $l\in \L$. Indeed, given $f'_1\circ\gamma_{l_1},f'_2\circ\gamma_{l_2}\in\A_0^\infty$, one takes the supremum $l$ of $l_1$ and $l_2$ (this corresponds to the coarsest lattice that is finer than both $l_1$ and $l_2$), and writes $f'_j\circ\gamma_{l_j}=(f'_j\circ\gamma_{l_jl})\circ\gamma_l\equiv f_j\circ\gamma_l$. The same goes for $k$ functions $f_1\circ\gamma_l,\ldots,f_k\circ\gamma_l$.
\end{rema} 
 
The first use of this remark is in defining a Poisson structure on $\A_0^\infty$. The Poisson bracket of $f_1\circ \gamma_l$ and $f_2\circ\gamma_l$ is defined as
\begin{align*}
	\{f_1\circ \gamma_l,f_2\circ\gamma_l\}:=\{f_1,f_2\}\circ\gamma_l,
\end{align*}
in terms of the Poisson bracket on $\A_0^l$, which is a Poisson subalgebra of $C^\infty(X^l)$. To show that the above bracket on $\A_0^\infty$ is well-defined, it suffices to show that $\{f_1\circ\gamma_{lm},f_2\circ\gamma_{lm}\}=\{f_1,f_2\}\circ\gamma_{lm}$ for all $l\leq m$. This follows from the analogous statement for $\gamma_{\text{add}}$ and $\gamma_{\text{sub}}$, which can be straightforwardly checked.
%, since we are in the case $G=\T^n$.

\subsection{The quantum systems and quantum embedding maps}
\label{sct:quantum systems and embedding maps}
To each lattice $l\in\L$ we will associate an operator system modeling the quantum system. This operator system is defined as a quantization of $\M_0^l$ under a quantization map $\Ql$ that defines an extension of Weyl quantization. Recall that every $f\in\M_0^l$ can be written as $f=\sum_k g_k \otimes \hat\mu_k$ for $g_k\in C^\infty(G^l)$ and $\supp(\mu_k)\subseteq B^l\subseteq\g^l$, where $B=B_{1/2}(0_\g)$. Notice that $\hbar\xi\in B^l$ for every $\hbar\in[-1,1]$ and $\xi\in\supp(\mu_k)$. % and that the exponential map restricted to $B^l\subseteq U^l$ is a diffeomorphism. 
We define the quantization map on the lattice $l$ to be
\begin{align}\label{eq:Ql}
	\Ql:\M_0^l&\to\mB(L^2(G^l)),\nonumber\\
	\Ql\bigg(\sum_{k=1}^Kg_k\otimes\hat\mu_k\bigg)\psi[y]&:=\sum_{k=1}^K \int_{\g^l} d\mu_k(\xi)g_k[y+\tfrac12\hbar\xi]\psi[y+\hbar\xi].
\end{align}

A simple calculation shows that, acting on the wave functions $\psi_a[x]:=e^{2\pi i a\cdot x}$ ($a\in(\Z^n)^l$), this quantization map has the simple form
\begin{align}\label{eq:asymptotically Ruben}
	\Ql(e^{2\pi i b\cdot}\otimes h)\psi_a=h(2\pi \hbar(a+\tfrac12 b))\psi_{a+b},
\end{align}
and therefore coincides with the one in \cite{vNS20}. Moreover, when $b$ is small enough, it coincides with Weyl quantization on the Riemannian manifold $\T^n$ as introduced in \cite[Definition 3.4.4]{Landsman98}, as the cut-off function $\kappa$ used there becomes 1 when we restrict to $\M_0^l$. The insight used by this paper is that, restricted to the operator system $\M_0^l$, the quantization map is injective. The quantum system associated to $l$ is defined by
	$$\M_\hbar^l:=\Ql(\M_0^l).$$
As $\Ql$ is linear, unital, and *-preserving, $\M_\hbar^l$ is an operator system.

\begin{exam}
A notable subset of $\M^l_0$ is $\mathfrak{W}^l_0:=\spn\{g\otimes e^{i\xi\cdot}:~g\in C^\infty(G^l),\xi\in  B^l\}$. This subset generates the C*-algebra $C(G^l)\hatotimes \W((\g^*)^l,0)$, which can be seen as a classical Weyl C*-algebra on the torus \cite{BHR,vN19,vNS20} that lies inside $A^l_0=\overline{\A_0^l}$. The image of $\mathfrak{W}^l_0$ under the above quantization map generates the crossed product C*-algebra $C(G^l)\rtimes G^l$. Indeed, we have $\Ql(g\otimes e^{i\xi\cdot})=M_{g\circ L_{[\hbar\xi/2]}}L^*_{[\hbar\xi]}$.
\end{exam}

\paragraph{Direct limit of Hilbert spaces}
To model the quantum system in infinite degrees of freedom, we will eventually construct a noncommutative C*-algebra that is canonically represented on a Hilbert space. This Hilbert space is the limit of the following direct system of Hilbert spaces:
\begin{align*}
	\quad\H^l:=L^2(G^l),\quad u^{ml}:=(\gamma^{conf}_{lm})^*:\H^l\to\H^m.
\end{align*}
Passing to the direct limit, we denote
\begin{align*}
	\H^{\infty}:=&\lim_{\rightarrow}\H^l,\quad
	u^l=(\gamma^{conf}_l)^*:\H^l\to\H^\infty.
\end{align*}
To define a direct limit of the operator systems $\M^l_\hbar$, we need to define the embedding maps and show that they satisfy the needed properties.

\paragraph{Quantum embedding maps}
The quantum embedding maps are defined by quantizing the classical embedding maps, i.e., for all $l\leq m\in\L$ and all $f\in\M_0^l$ we define
\begin{align*}
	F^{ml}_Q:\M^l_\hbar&\to\M^m_\hbar,\\
	F^{ml}_Q(\Ql(f))&:=\Qm(F^{ml}_C(f)),
\end{align*}
which is unambiguous by injectivity of $\Ql$.
\begin{exam}
	The embedding map $F_Q^\add$ is given by tensoring with $1$, which exemplifies why our quantum systems should be unital. The embedding map $F_Q^\sub$ is best understood on elements of %(an operator system within) 
	$C(G^l)\rtimes G^l.$ As depicted in Figure \ref{fig:QE1}, we have
	$$F_Q^\sub(M_gL^*_{[\xi]})=M_{g\circ \mu}L^*_{\left[\frac{d_1}{d}\xi,\frac{d_2}{d}\xi\right]},$$
where $g\in C(G)$, $\xi\in B^l$ and $\mu:\T^n\times\T^n\to\T^n$ is given by $\mu([x_1],[x_2]):=[x_1+x_2]$. One sees the metric at work, as well as the exponential map $[\cdot]:\g\to G$, and notices that the well-definedness of the UV-limit hinges on the use of operator systems.
\end{exam}
%\co{Zie p.421}

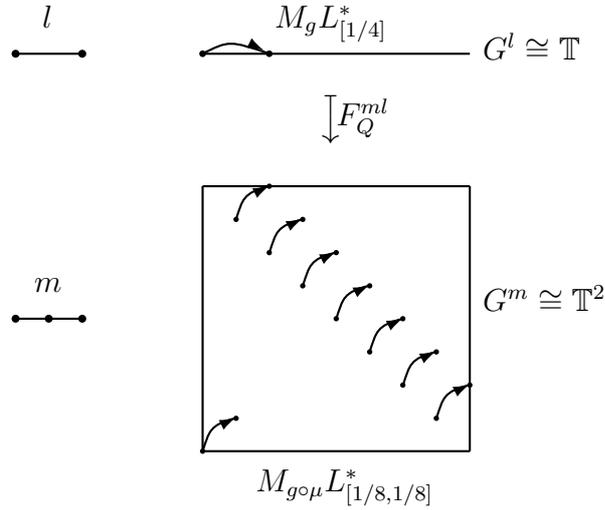
\begin{figure}[!htb]
	\centering
	\thicklines
	\begin{picture}(150,180)(-50,-20)
		\put(101,148){ $G^l\cong\T$}	
		\put(101,53){ $G^m\cong\T^2$}
		
	%roosters l, m
		\put(-70,150){\line(1,0){25}}
		\put(-70,150){\circle*{3}}
		\put(-45,150){\circle*{3}}
		\put(-60,160){$l$}
		\put(-70,50){\line(1,0){25}}
		\put(-70,50){\circle*{3}}
		\put(-57.5,50){\circle*{3}}
		\put(-45,50){\circle*{3}}
		\put(-63,60){$m$}

	%vlakke cirkel
		\put(0,150){\line(1,0){100}}

	%vlakke torus
		\put(0,0){\line(1,0){100}}
		\put(0,0){\line(0,1){100}}
		\put(100,0){\line(0,1){100}}	
		\put(0,100){\line(1,0){100}}
			
	%mapsto
		\put(45,135){\rotatebox{270}{$\longmapsto$}}
		\put(50,123){$F^{ml}_Q$}		
		
	%vectoren cirkel
		\put(0,150){\circle*{2.5}}
		\put(25,150){\circle*{2.5}}
		\cbezier(0,150)(10,155)(15,155)(25,150)
		\put(20.5,152.5){\vector(16,-8){2.5}}
		
	%vectoren torus
		\put(0,0){\circle*{2}}
		\put(12.5,12.5){\circle*{2}}
		\cbezier(0,0)(2.5,7.5)(5,10)(12.5,12.5)	
		\put(10,11.5){\vector(2,1){2.5}}
		
		\put(87.5,12.5){\circle*{2}}
		\put(100,25){\circle*{2}}
		\cbezier(87.5,12.5)(90,20)(92.5,22.5)(100,25)	
		\put(97.5,24){\vector(2,1){2.5}}
		
		\put(75,25){\circle*{2}}
		\put(87.5,37.5){\circle*{2}}
		\cbezier(75,25)(77.5,32.5)(80,35)(87.5,37.5)	
		\put(85,36.5){\vector(2,1){2.5}}
		
		\put(62.5,37.5){\circle*{2}}
		\put(75,50){\circle*{2}}
		\cbezier(62.5,37.5)(65,45)(67.5,47.5)(75,50)	
		\put(72.5,49){\vector(2,1){2.5}}
		
		\put(50,50){\circle*{2}}
		\put(62.5,62.5){\circle*{2}}
		\cbezier(50,50)(52.5,57.5)(55,60)(62.5,62.5)	
		\put(60,61.5){\vector(2,1){2.5}}
		
		\put(37.5,62.5){\circle*{2}}
		\put(50,75){\circle*{2}}
		\cbezier(37.5,62.5)(40,70)(42.5,72.5)(50,75)	
		\put(47.5,74){\vector(2,1){2.5}}
		
		\put(25,75){\circle*{2}}
		\put(37.5,87.5){\circle*{2}}
		\cbezier(25,75)(27.5,82.5)(30,85)(37.5,87.5)	
		\put(35,86.5){\vector(2,1){2.5}}
		
		\put(12.5,87.5){\circle*{2}}
		\put(25,100){\circle*{2}}
		\cbezier(12.5,87.5)(15,95)(17.5,97.5)(25,100)	
		\put(22.5,99){\vector(2,1){2.5}}
		
		\put(27,160){$M_gL^*_{[1/4]}$}
		\put(20,-15){$M_{g\circ\mu} L^*_{[1/8,1/8]}$}		
	\end{picture}	
	\caption{A pictoral representation of an operator $M_gL^*_{[1/4]}\in\mathcal{M}^l_\hbar$ and its image under the quantum embedding map, where $G=\T$, $l$ has one edge, and $m=l^2$. For the picture, $g$ is supported closely around $[1/4]\in G^l$. The embedding map clearly respects the gauge action coming from the central vertex of $m$.}
	\label{fig:QE1}
\end{figure}		

Our quantum embedding maps contrast with those used in the existing literature on C*-algebraic lattice gauge theory \cite{ASS,BG,BS19a,BS19b,Stienstra,ST} because ours do not define a direct system (inductive system) of *-algebras. They therefore warrant some motivation.

We assume the situation of Figure \ref{fig:QE1} and Figure \ref{fig:QE2}, where $G=\T=U(1)$ and a lattice $l$ consisting of a single edge is compared to a lattice $m\geq l$ with two edges. 
There exist multiple observables on the lattice $m$ that have the same behavior when restricted to $l$. This can be seen in Figure \ref{fig:QE2}, in the formulas, or by interpreting the gauge field as rigid rotors associated to every edge, as in \cite{KS}. 
Indeed, the two rotors associated to the two edges of $m$ can either both be turned clockwise by a quarter circle, or both anti-clockwise by a quarter circle. When describing the gauge field by a single rotor, the two operations appear as the same observable (see Figure \ref{fig:QE2}).

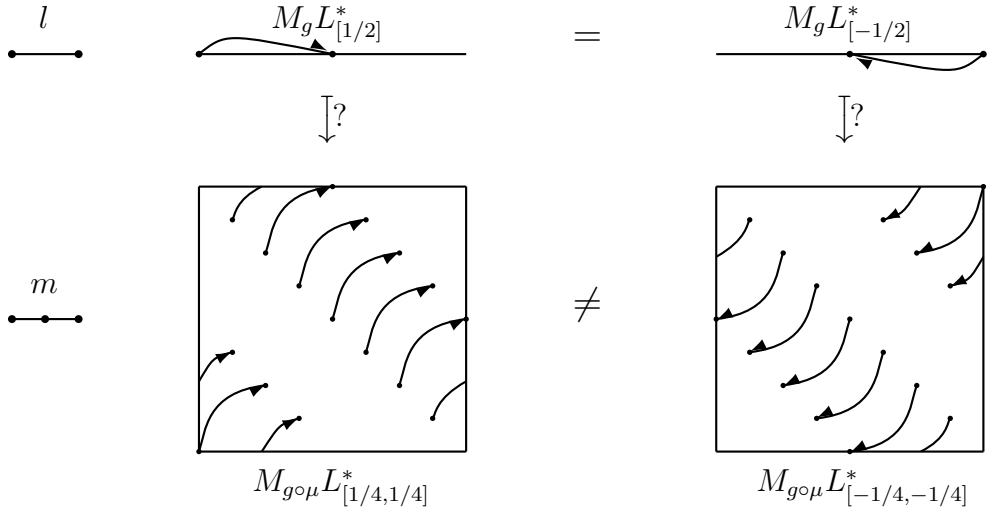
\begin{figure}[!htb]
	\centering
	\thicklines
	\begin{subfigure}{.49\textwidth}
	\centering
	\begin{picture}(135,200)(-50,-20)
	%roosters l, m
		\put(-70,150){\line(1,0){25}}
		\put(-70,150){\circle*{3}}
		\put(-45,150){\circle*{3}}
		\put(-60,160){$l$}
		\put(-70,50){\line(1,0){25}}
		\put(-70,50){\circle*{3}}
		\put(-57.5,50){\circle*{3}}
		\put(-45,50){\circle*{3}}
		\put(-63,60){$m$}

	%vlakke cirkel
		\put(0,150){\line(1,0){100}}
		
	%vectoren cirkel
		\put(0,150){\circle*{2.5}}
		\cbezier(0,150)(10,158)(40,158)(50,150)
		\put(45.5,152.5){\vector(16,-8){2.5}}
		\put(50,150){\circle*{2.5}}

	%vlakke torus
		\put(0,0){\line(1,0){100}}
		\put(0,0){\line(0,1){100}}
		\put(100,0){\line(0,1){100}}	
		\put(0,100){\line(1,0){100}}
			
	%mapsto
		\put(45,135){\rotatebox{270}{$\longmapsto$}}
		\put(50,123){$?$}	

	%vectoren torus
		\put(0,0){\circle*{2}}
		\put(25,25){\circle*{2}}
		\cbezier(0,0)(2.5,7.5)(17.5,22.5)(25,25)	
		\put(22.5,24){\vector(2,1){2.5}}
		
		\put(75,25){\circle*{2}}
		\put(100,50){\circle*{2}}
		\cbezier(75,25)(77.5,32.5)(92.5,47.5)(100,50)	
		\put(97.5,49){\vector(2,1){2.5}}
		
		\put(62.5,37.5){\circle*{2}}
		\put(87.5,62.5){\circle*{2}}
		\cbezier(62.5,37.5)(65,45)(80,60)(87.5,62.5)	
		\put(85,61.5){\vector(2,1){2.5}}

		\put(50,50){\circle*{2}}
		\put(75,75){\circle*{2}}
		\cbezier(50,50)(52.5,57.5)(67.5,72.5)(75,75)	
		\put(72.5,74){\vector(2,1){2.5}}
		
		\put(37.5,62.5){\circle*{2}}
		\put(62.5,87.5){\circle*{2}}
		\cbezier(37.5,62.5)(40,70)(55,85)(62.5,87.5)	
		\put(60,86.5){\vector(2,1){2.5}}
		
		\put(25,75){\circle*{2}}
		\put(50,100){\circle*{2}}
		\cbezier(25,75)(27.5,82.5)(42.5,97.5)(50,100)	
		\put(47.5,99){\vector(2,1){2.5}}
		
	%Halve vectoren
		\put(12.5,87.5){\circle*{2}}
		\cbezier(12.5,87.5)(13.5,89)(15,95)(23.5,100)

		\put(37.5,12.5){\circle*{2}}
		\cbezier(23.5,0)(30,10)(35,10.5)(37.5,12.5)
		\put(35,11.5){\vector(2,1){2.5}}

		\put(12.5,37.5){\circle*{2}}
		\cbezier(0,26.5)(5,35)(10,35.5)(12.5,37.5)
		\put(10,36.5){\vector(2,1){2.5}}
		
		\put(87.5,12.5){\circle*{2}}
		\cbezier(87.5,12.5)(88.5,14)(90,20)(100,26.5)
		
		\put(140,153){\large $=$}
		\put(140,52){\large $\neq$}
		
		\put(27,160){$M_gL^*_{[1/2]}$}
		\put(20,-15){$M_{g\circ\mu} L^*_{[1/4,1/4]}$}		
	\end{picture}
	\end{subfigure}
	\begin{subfigure}{.49\textwidth}
	\centering
	\begin{picture}(225,200)(-65,-20)

	%vlakke cirkel
		\put(0,150){\line(1,0){100}}
		
	%vectoren cirkel
		\put(100,150){\circle*{2.5}}
		\cbezier(100,150)(90,142)(60,142)(50,150)
		\put(54.5,147.5){\vector(-16,8){2.5}}
		\put(50,150){\circle*{2.5}}

	%vlakke torus
		\put(0,0){\line(1,0){100}}
		\put(0,0){\line(0,1){100}}
		\put(100,0){\line(0,1){100}}	
		\put(0,100){\line(1,0){100}}			
		
	%mapsto
		\put(45,135){\rotatebox{270}{$\longmapsto$}}
		\put(50,123){$?$}		
					
	%vectoren torus
		\put(100,100){\circle*{2}}
		\put(75,75){\circle*{2}}
		\cbezier(100,100)(97.5,92.5)(82.5,77.5)(75,75)	
		\put(77.5,76){\vector(-2,-1){2.5}}
		
		\put(25,75){\circle*{2}}
		\put(0,50){\circle*{2}}
		\cbezier(25,75)(22.5,67.5)(7.5,52.5)(0,50)	
		\put(2.5,51){\vector(-2,-1){2.5}}

		\put(37.5,62.5){\circle*{2}}
		\put(12.5,37.5){\circle*{2}}
		\cbezier(37.5,62.5)(35,55)(20,40)(12.5,37.5)	
		\put(15,39){\vector(-2,-1){2.5}}
		
		\put(50,50){\circle*{2}}
		\put(25,25){\circle*{2}}
		\cbezier(50,50)(47.5,42.5)(32.5,27.5)(25,25)	
		\put(27.5,26){\vector(-2,-1){2.5}}

		\put(62.5,37.5){\circle*{2}}
		\put(37.5,12.5){\circle*{2}}
		\cbezier(62.5,37.5)(60,30)(45,15)(37.5,12.5)	
		\put(40,14){\vector(-2,-1){2.5}}
		
		\put(75,25){\circle*{2}}
		\put(50,0){\circle*{2}}
		\cbezier(75,25)(72.5,17.5)(57.5,2.5)(50,0)	
		\put(52.5,1){\vector(-2,-1){2.5}}
	
	%Halve vectoren
		\put(87.5,12.5){\circle*{2}}
		\cbezier(87.5,12.5)(86.5,11)(85,5)(76.5,0)

		\put(62.5,87.5){\circle*{2}}
		\cbezier(76.5,100)(70,90)(65,89.5)(62.5,87.5)
		\put(65,88.5){\vector(-2,-1){2.5}}

		\put(87.5,62.5){\circle*{2}}
		\cbezier(100,73.5)(95,65)(90,64.5)(87.5,62.5)
		\put(90,63.5){\vector(-2,-1){2.5}}
		
		\put(12.5,87.5){\circle*{2}}
		\cbezier(12.5,87.5)(11.5,86)(10,80)(0,73.5)
%		\put(70,223){\vector(2,1){5}}

		\put(25,160){$M_gL^*_{[-1/2]}$}
		\put(15,-15){$M_{g\circ\mu} L^*_{[-1/4,-1/4]}$}		
	\end{picture}
	\end{subfigure}
	\caption{The quantum embedding map does not extend in a multiplicative way from $\mathcal{M}^l_\hbar$ to the algebra $\A^l_\hbar$ generated by $\mathcal{M}^l_\hbar$. If we would try, we would end up with two representations of the same observable in $\mathcal A^l_\hbar$ being mapped to two different observables in $\mathcal A^m_\hbar$.}
	\label{fig:QE2}
\end{figure}

 Therefore, if one wants to interpret an observable on a lattice $l$ as an observable on the continuum, a choice has to be made. We make this choice by restricting at any finite level to observables that rotate any rotor less than a certain amount, so that an embedding of such an observable can be made by fairly distributing that rotation over the smaller rotors that make up the original one. Clearly, this means giving up on multiplicative structure. This is not against the C*-algebraic philosophy, however,
which states that one can describe any physical system once we have a sufficiently rich C*-algebra of observables. The set of observables at a finite level makes up but a subset of the full algebra, and is therefore not required to completely describe a physical system. Only the full set of observables, with arbitrary lattice size, can discern between any two gauge fields, and can therefore be expected to form a *-algebra (lying densely in a C*-algebra). That is indeed what we will prove in Proposition \ref{prop:algebra}.

%A classical or quantum observable in $\A^l_0$ or $\A^l_\hbar$ on a lattice $l$ corresponds to a function or operator that acts on gauge fields in a special way; it only uses holonomic information of the gauge fields restricted to $l$. As a consequence, the whole space of observables associated to $l$ can not discern between gauge fields that, in the direction of an edge of $l$, wrap around the circle arbitrarily many times. The distinction can only be made by including observables on finer lattices. Therefore, the concatenation of two observables on the same lattice, does not have to have a meaningful interpratation as an observable on the lattice $l$, but, luckily, can always be identified with an observable on a finer lattice $m$. Multiplication, we therefore argue, has to be released at the finite level, in order to be retrieved in the continuum limit. An example of the failure of multiplicativity of the embedding map is shown in Figure \ref{fig:QE2}. It is also clear from Figure \ref{fig:QE2} that, when embedding an observable into a finer lattice $m\geq l$, one has to make some choice. The choice reflects the UV-cutoff that is given by lattice, and should not be of influence to the continuum C*-algebra. The rest of this paper shall found these considerations on strict mathematical results, in particular showing that the quantum embedding maps define an algebra in the continuum, and that there the quantization map on which the quantum embedding maps are based forms a strict deformation quantization.

As further motivation of our quantum embedding maps, and to be used later, we show that they intertwine the direct system of Hilbert spaces given by $u^{ml}:\H^l\to\H^m$.

%\subsubsection{Embedding maps intertwine}
%We now show that the embedding maps $F^{ml}_Q$ respect the embedding maps $u^{ml}$.
\begin{lem}\label{lem:F_Q and u}
	For $l\leq m\in\L$ and $O\in\M_\hbar^l$ we have
\begin{align*}
	F_Q^{ml}(O)u^{ml}=u^{ml}O.
\end{align*}
\end{lem}
\begin{proof}
	Similar to \eqref{S^ml}, we define
	\begin{align*}
		T^\add(\xi):=(\xi,0)\qquad T^\sub(\xi):=(\xi,\xi),
	\end{align*}	
	to obtain a direct system of linear maps $T^{ml}:\g^l\to\g^m$. We account here that
	\begin{align}\label{eq:S and T}
		\gamma^\mom_{lm}(X)\cdot\xi=X\cdot S^{ml}(\xi);\qquad \gamma^\conf_{lm}(q)\cdot \xi=q\cdot T^{ml}(\xi);\qquad \gamma_{lm}^\mom\circ T^{ml}=\id_{\g^l},
	\end{align}
	such that, in particular, $u^{ml}\psi_a=\psi_{T^{ml}(a)}$ for all $a\in(\Z^n)^l$. For $f=e^{2\pi ib\cdot}\otimes h$, we get
	\begin{align*}
		\Qm(f\circ\gamma_{lm})u^{ml}\psi_a&=\Qm(f\circ\gamma_{lm})\psi_{T^{ml}(a)}\\
		&=h(\gamma^\mom_{lm}(2\pi\hbar(T^{ml}(a)+\tfrac12 T^{ml}(b))))\psi_{T^{ml}(a)+T^{ml}(b)}\\
		&=h(2\pi\hbar(a+\tfrac12 b))u^{ml}\psi_{a+b},
	\end{align*}
	so $\Qm(f\circ\gamma_{lm})u^{ml}\psi_a=u^{ml}\Ql(f)\psi_a$, which implies the lemma.
\end{proof}

\subsection{The continuum quantization map and quantum system}

To define $\Q$, we define $\Q(f\circ\gamma_l)\in\mB(\H^\infty)$ by its action on $u^m\psi\in\H^\infty$, where $m\geq l$, namely
\begin{align*}
	\Q(f\circ\gamma_l)u^m\psi:=u^mQ^m_\hbar(f\circ\gamma_{lm})\psi\qquad (\psi\in\H^m).
\end{align*}
To show that this is well defined, we use Lemma \ref{lem:F_Q and u} and find, for all $n\geq m\geq l$ and $\psi\in\H^m$,
\begin{align*}
	u^n\Qn(f\circ\gamma_{ln})u^{nm}\psi &= u^n F_Q^{nm}(\Qm(f\circ\gamma_{lm}))u^{nm}\psi\\
	&= u^m\Qm(f\circ\gamma_{lm})\psi,
\end{align*}
and conclude that $\Q(f\circ\gamma_i)$ is well-defined on the dense subset $\cup_m u^m\H^m\subseteq\H^\infty$. If we write $f=\sum_k g_k\otimes \hat\mu_k$ we obtain, %for $\mu_k\in C^\infty(G^l)$, $h_k\in \Wr$ we obtain
\begin{align*}
	\quadnorm{\Q(f\circ\gamma_l)u^m\psi} &= \quadnorm{\Qm(f\circ\gamma_{lm})\psi}\leq \sum_k\supnorm{g_k}\absnorm{\mu_k}\quadnorm{u^m\psi}. %= \sum_k\supnorm{g_k}\absnorm{\mu_k}\quadnorm{u^m\psi}.
\end{align*}
Therefore $\Q:\A_0^\infty\to\mB(\H^\infty)$ is well defined, and $\norm{\Q(f\circ\gamma_l)}\leq\sum\supnorm{g_k}\absnorm{\mu_k}$, independently from $\hbar$. The above also shows that 
\begin{align}
	\norm{\Q(f\circ\gamma_l)}=\sup_{m\geq l}\norm{\Qm(f\circ\gamma_{lm})}=\lim_{m}\norm{\Qm(f\circ\gamma_{lm})}.\label{eq:Q^infty sup over Q^j}
\end{align}
We define $$\A^\infty_\hbar:=\Q(\A^\infty_0)\equiv \{F^l_Q(a):~l\in\L,~ a\in\M_\hbar^l\}.$$ We write $\A_\hbar^\infty$ instead of $\M_\hbar^\infty$ to suggest it is in fact an algebra.

\begin{prop}\label{prop:algebra}
	The operator system $\A_\hbar^\infty=\Q(\A_0^\infty)$ is a *-algebra.
\end{prop}
\begin{proof}
	By Remark \ref{rema:supremum trick}, we only have to show that $\Q(f_1\circ\gamma_l)\Q(f_2\circ\gamma_l)$ is in $\A_\hbar^\infty$. Write $O_1:=\Ql(f_1)$ and $O_2:=\Ql(f_2)$. Because we cannot take their product in the operator system $\M^l_\hbar$, we first subdivide the edges of $l$ to obtain the lattice $l^2$ defined by Definition \ref{def:l_R}. A straightforward computation shows firstly that
		\begin{align*}
			F_Q^{l^2l}(O_1)F_Q^{l^2l}(O_2)\in\M_\hbar^{l^2},
		\end{align*}
		and secondly that
		\begin{align*}
			F_Q^{ml}(O_1)F_Q^{ml}(O_2)=F_Q^{ml^2}(F_Q^{l^2l}(O_1)F_Q^{l^2l}(O_2)),
		\end{align*}
		for all $m\geq l^2$.
		Using this formula and Lemma \ref{lem:F_Q and u}, we obtain
		\begin{align*}
			\Q(f_1\circ\gamma_l)\Q(f_2\circ\gamma_l)u^m\psi&=u^mF_Q^{ml}(O_1)F_Q^{ml}(O_2)\psi\\
			&=F_Q^{l^2}(F_Q^{l^2l}(O_1)F_Q^{l^2l}(O_2))u^m\psi,
		\end{align*}
		for all $u^m\psi\in\H^\infty$. Hence, $\Q(f_1\circ\gamma_l)\Q(f_2\circ\gamma_l)=F_Q^{l^2}(F_Q^{l^2l}(O_1)F_Q^{l^2l}(O_2))\in\A_\hbar^\infty$.
\end{proof}

Taking the closures of $\A_0^\infty\subseteq C_b(X^\infty)$ and $\A_\hbar^\infty\subseteq\mB(\H^\infty)$, we therefore obtain C*-algebras $A_0^\infty$ and $A_\hbar^\infty$. By Theorem \ref{thm:sdq infinity}, we are justified in saying that the noncommutative C*-algebra $A_\hbar^\infty$ is obtained by strict deformation quantization of $A_0^\infty$.

\subsection{Time evolution}\label{sct:time evolution}
Before moving on to strict deformation quantization, we state two promising results with respect to time evolution. They show that our C*-algebras are invariant under what one could call `free time evolution' in the continuum limit.
On the finite level, these results can be strengthened to invariance under \textit{all} time evolutions. This is proven in \cite[Theorem 15]{vNS20} and \cite[Theorem 29]{vNS20}.
The combination of the results here and in \cite{vNS20} indicates that we are on the right track to obtaining classical and quantum C*-algebras that are invariant under respectively classical and quantum Yang--Mills time evolution, as discussed in Section \ref{sct:outlook}.
\begin{thm}\label{thm:classical time evolution infty}
	The C*-algebra $A^\infty_0\subseteq C_b(X^\infty)$ is conserved by the time evolution given on a lattice $l\in\L$ by the Hamiltonian $H_l:T^*G^l\to\R$, $H_l(q,v):=\sum_{e\in l}d_ev_e^2$.
\end{thm}
\begin{proof}
	Every Hamiltonian $H_l$ induces a time-evolution $\tau^0_{l}: \R\times A_0^l\to A_0^l$ by \cite[Lemma 10]{vNS20}. It can be checked that $H_l\circ\gamma_{lm}=H_m$, and therefore $\tau_m^0(t,f\circ\gamma_{lm})=\tau_l^0(t,f)\circ\gamma_{lm}$. We conclude that the time-evolution
	\begin{align*}
		\tau_\infty^0:A_0^\infty\to A_0^\infty,\qquad\tau_\infty^0(t,f\circ\gamma_l):=\tau_l^0(t,f)\circ\gamma_l
	\end{align*}
	is well-defined.
\end{proof}

\begin{thm}\label{thm:quantum time evolution infty}
	The C*-algebra $A^\infty_\hbar\subseteq\mB(\H^\infty)$ is conserved by the time evolution given on a lattice $l\in\L$ by the Hamiltonian $\hat{H}_l:=\sum_{e\in l}d_e\partial_e^2$.
\end{thm}
\begin{proof}
	These Hamiltonians define a continuum Hamiltonian $\hat H_\infty$ in $\H^\infty$ with domain %$\cup_l u^l(\dom H_l)$.
	\begin{align*}
		\dom \hat H_\infty:=\bigcup_{l\in\L} u^l(\dom \hat H_l)=\bigcup_{l\in\L} u^l(C^\infty(G^l)),
	\end{align*}
	namely $\hat{H}_\infty u^l\psi:=u^l\hat H_l\psi$. Straightforwardly, one checks well-definedness and essential self-adjointness. By \cite[Remark 27]{vNS20}, we have
	\begin{align*}
		e^{it\hat H_\infty}\Q(f\circ\gamma_l)e^{-it\hat H_\infty}u^m\psi&=u^me^{it\hat H_m}\Qm(f\circ\gamma_{lm})e^{-it\hat H_m}\psi\\
		&=u^m\Qm(\tau^0_m(t,f\circ\gamma_{lm}))\psi\\
		&=u^m\Qm(\tau^0_l(t,f)\circ\gamma_{lm})\psi\\
		&=\Q(\tau^0_l(t,f)\circ\gamma_l)u^m\psi.
	\end{align*}
	Therefore, $e^{it\hat H_\infty}\Q(f\circ\gamma_l)e^{-it\hat H_\infty}=\Q(\tau^0_l(t,f)\circ\gamma_l)\in A_\hbar^\infty$ for every $t$.
\end{proof}
%\begin{rema}
%	The above proof is specific for the electric (kinetic) part of time evolution. The full dynamics (in the sense of \cite{KS}) is expected to require a more complicated argument, as quantization does not commute with time evolution in general. The author would like to suggest a different route towards full time evolution, namely to generalize \cite[Theorem 29]{vNS20} to the continuum limit. However, this might be complicated, and is beyond the scope of this paper.
%\end{rema}

\section{Strict deformation quantization}
\label{sct:SDQ}

In this section we prove our main theorem, which is formulated as follows.

\begin{thm}\label{thm:sdq infinity}
	Let $Q_0^\infty:=\id_{\A^\infty_0}$. The maps $\Q:\A_0^\infty\to A^\infty_\hbar$ for $\hbar\in I:=[-1,1]$ form a strict deformation quantization. That is, $\Q$ is a *-preserving injective linear map whose image is an algebra, and for all $f,g\in\A_0^\infty$ it holds that
	\begin{align*}
		\lim_{\hbar\to0}\norm{\Q(f)\Q(g)-\Q(fg)}=0&\qquad\text{(von Neumann's condition);}\\				\lim_{\hbar\to0}\norm{(-i\hbar)^{-1}[\Q(f),\Q(g)]-\Q(\{f,g\})}=0&\qquad\text{(Dirac's condition)};\\
		\text{the map}\quad I\to\R,\quad\hbar\mapsto\norm{\Q(f)}\quad\text{is continuous}&\qquad\text{(Rieffel's condition).}
	\end{align*}
\end{thm}

For readability, the proof of Theorem \ref{thm:sdq infinity} is split up into Propositions \ref{prop:star-preserving}, \ref{prop:injective}, \ref{prop:von Neumann}, \ref{prop:Dirac}, \ref{prop:Rieffel0}, and \ref{prop:Rieffel1}.

\begin{prop}\label{prop:star-preserving}
	The map $\Q:\A_0^\infty\to\A_\hbar^\infty$ is linear and *-preserving for all $\hbar\in I$.
\end{prop}
\begin{proof}
	Linearity is obvious, so we are left to prove that $\Q(f)^*=\Q(\overline{f})$ for $f\in\A_0^\infty$. Given $f\circ\gamma_l\in\A_0^\infty$ and $u^m\psi^m,u^n\psi^n\in\H^\infty$, choose $p\geq l,m,n$. By using that $\Ql:\A_0^l\to \A_\hbar^l$ is star-preserving (which can be derived from \cite[Propostion 18(1)]{vNS20}, or directly from \eqref{eq:Ql}), we get
		\begin{align*}
			\p{\Q(f\circ\gamma_l)u^m\psi^m}{u^n\psi^n}
			&=\p{u^pQ_\hbar^p(f\circ \gamma_{lp})u^{pm}\psi^m}{u^pu^{pn}\psi^n}\\
			&=\p{\Qp(f\circ\gamma_{lp})u^{pm}\psi^m}{u^{pn}\psi^n}\\
			&=\p{u^{pm}\psi^m}{\Qp(\overline{f}\circ\gamma_{lp})u^{pn}\psi^n}\\
			&= \p{u^m\psi^m}{\Q(\overline{f}\circ\gamma_l)u^n\psi^n}.
		\end{align*}
		Therefore $\Q(f\circ\gamma_l)^*$ equals $\Q(\overline{f}\circ\gamma_l)$ on a dense subset of $\H^\infty$, hence on the whole of $\H^\infty$ by boundedness.
\end{proof}

\begin{prop}\label{prop:injective}
	The map $\Q:\A_0^\infty\to A_\hbar^\infty$ is injective for all $\hbar\in I$.
\end{prop}
\begin{proof}
	Suppose $\Q(f\circ\gamma_l)=0$ for some $f\in\M_0^l$. Then
		\begin{align*}
			0&=\Q(f\circ\gamma_l)u^l\psi=u^l\Ql(f)\psi
		\end{align*}
		for all $\psi\in\H^l$. So $0=\Ql(f)=\sum_k\int d\mu_k(\xi) M_{g_k(\cdot+ \hbar\xi/2)}L^*_{[\hbar\xi]}$. Since $|\hbar|\leq 1$, we find that we always have $\hbar\xi\in B^l$ under the integral, and must have $f=0$.
\end{proof}

\begin{prop}\label{prop:von Neumann}
	\emph{(von Neumann's condition)} For each $f,g\in\A_0^\infty$, we have
			\begin{align*}
				\lim_{\hbar\to0}\norm{\Q(f)\Q(g)-\Q(fg)}=0.
			\end{align*}
\end{prop}
\begin{proof}
	The proof is based on that of \cite[Theorem 22(2)]{vNS20}, but more complicated because $\Q$ is defined on $\H^\infty$, which includes all $\H^m$. Therefore, estimating an operator norm in $\mB(\H^\infty)$ amounts to taking a supremum over $m$. For two lattices $l\leq m\in\L$ and a function $e^{2\pi i b\cdot}\otimes h\in\M_0^l$ we have,
		\begin{align*}
			(e^{2\pi i b\cdot}\otimes h)\circ\gamma_{lm}=e^{2\pi ib\cdot \gamma^{\text{conf}}_{lm}(\cdot)}\otimes h\circ\gamma_{lm}^{\text{mom}}
			=e^{2\pi iT^{ml}(b)\cdot}\otimes h\circ\gamma_{lm}^{\text{mom}}
		\end{align*}
		where we used \eqref{eq:S and T}. Combining \eqref{eq:asymptotically Ruben} with \eqref{eq:S and T}, we find
		\begin{align*}
			\Qm((e^{2\pi i b\cdot}\otimes h)\circ\gamma_{lm})\psi_a&=h(\gamma^\mom_{lm}(2\pi\hbar(a+\tfrac{1}{2}T^{ml}(b))))\psi_{a+T^{ml}(b)}\\
			&=h(2\pi\hbar(\gamma^\mom_{lm}(a)+\tfrac{1}{2}b))\psi_{a+T^{ml}(b)}\,.
		\end{align*}
		Fix $f_1=e^{2\pi ib_1\cdot}\otimes h_1,f_2=e^{2\pi ib_2\cdot}\otimes h_2\in\M_0^l$ for an $l\in\L$. By bilinearity and Remark \ref{rema:supremum trick} it suffices to prove the proposition for $f=f_1\circ\gamma_l$ and $g=f_2\circ\gamma_l$. We note that if $O\psi_a=F(a)\psi_{a+b}$ for some $F\in C_b((\Z^n)^m)$ and $b\in(\Z^n)^m$, then clearly $\norm{O}=\sup_{a\in(\Z^n)^m}\quadnorm{O\psi_a}$. We find, %by following the steps in the proof of \cite[Theorem 22(2)]{vnuland20},
		\begin{align*}
			&\sup_{m\in\L_{\geq l}}\sup_{a\in(\Z^n)^{m}}\quadnorm{\left(\Qm(f_1f_2\circ\gamma_{lm})-\Qm(f_1\circ\gamma_{lm})\Qm(f_2\circ\gamma_{lm})\right)\psi_a}\\
			&\quad\leq \sup_{m\in\L_{\geq l}}\sup_{a\in(\Z^n)^m}\Big|h_1(2\pi\hbar(\gamma_{lm}^\mom(a)+\tfrac{1}{2}b_1+\tfrac{1}{2}b_2))h_2(2\pi\hbar(\gamma_{lm}(a)+\tfrac{1}{2}b_1+\tfrac{1}{2}b_2))\\
			&\qquad-h_1(2\pi\hbar(\gamma_{lm}^\mom(a)+\tfrac{1}{2}b_1+b_2))h_2(2\pi\hbar(\gamma^\mom_{lm}(a)+\tfrac{1}{2}b_2))\Big|\\
			&\quad\leq \supnorm{h_1}\pi|\hbar|\supnorm{\partial_{b_1} h_2} + \pi|\hbar|\supnorm{\partial_{b_2} h_1}\supnorm{h_2}\to 0\quad(\hbar\to0),
		\end{align*}
		which by \eqref{eq:Q^infty sup over Q^j} completes the proof.
\end{proof}

\begin{prop}\label{prop:Dirac}
	\emph{(Dirac's condition)} For each $f,g\in\A_0^\infty$, we have
			\begin{align*}
				\lim_{\hbar\to0}\norm{(-i\hbar)^{-1}[\Q(f),\Q(g)]-\Q(\{f,g\})}=0.
			\end{align*}
\end{prop}
\begin{proof}
	Similar to the proof of Proposition \ref{prop:von Neumann}, % now by using the proof of \cite[Theorem 22(3)]{vnuland20}, 
	we obtain
		\begin{align*}
			&\sup_{m\in\L_{\geq l}}\sup_{a\in(\Z^n)^{m}}\norm{\left(\frac{i}{\hbar}[\Qm(f_1\circ\gamma_{lm}),\Qm(f_2\circ\gamma_{lm})]-\Qm(\{f_1\circ\gamma_{lm},f_2\circ\gamma_{lm}\})\right)\psi_a}\\
			&\quad\leq \sup_{m\in\L_{\geq l}}\sup_{a\in(\Z^n)^{m}}\Big|\frac{i}{\hbar}\Big(h_1\big(2\pi\hbar(\gamma_{lm}^\mom(a)+b_2+\tfrac12 b_1)\big)h_2\big(2\pi\hbar(\gamma_{lm}^\mom(a)+\tfrac12 b_2)\big)\\
			&\qquad\quad-h_1\big(2\pi\hbar(\gamma_{lm}^\mom(a)+\tfrac12 b_1)\big)h_2\big(2\pi\hbar(\gamma_{lm}^\mom(a)+b_1+\tfrac12 b_2)\big)\Big)\\
			&\qquad -2\pi i\Big(\partial_{b_2}h_1\cdot h_2-h_1\cdot\partial_{b_1}h_2\Big)\Big(2\pi\hbar(\gamma_{lm}^\mom(a)+\tfrac12(b_1+b_2))\Big)\Big|\\
			&\quad\to0\quad(\hbar\to0),
		\end{align*}
		which by \eqref{eq:Q^infty sup over Q^j} completes the proof.
\end{proof}

\subsection{Rieffel's condition at zero}

Rieffel's condition is all that remains to prove in order to establish our main theorem. Its proof is by far the most difficult part of this paper, and is split into two parts, the first part giving continuity around $\hbar=0$ and the second part giving continuity elsewhere. 

For the first part we will use the following lemma.

\begin{lem}\label{lem:supnorm from operators}
	Let $f=\sum_{k=1}^Kg_k\otimes \hat\mu_k\in\M_0^l$ for $l\in\L$. 
	For every $m\geq l$, we have
	\begin{align*}
		\supnorm{f}=\supnorm{F_C^{ml}(f)}=\sup_{q\in G^m}\norm{\sum_{k=1}^K\int_{\g^l}d\mu_k(\xi)g_k(\gamma_{lm}^\conf(q))L^*_{S^{ml}(\hbar\xi)}},
	\end{align*}
	where, on the right hand side, the norm is the operator norm on $\mB(L^2(\g^m))$ and the integral is interpreted strongly.
\end{lem}
\begin{proof}
	The first equality is immediate, as $F^{ml}_C=(\gamma_{lm})^*$ and $\gamma_{lm}$ is surjective. By \eqref{eq:F_C^ml in termen van S^ml}, it now suffices to prove the lemma in the case where $l=m$, so $\gamma_{lm}^\conf=\id$ and $S^{ml}=\id$.	We obtain
	\begin{align*}
		\supnorm{f}&=\sup_{q\in G^l}\supnorm{\sum g_k(q)\int d\mu_k(\xi)e^{i\hbar\xi\cdot}}\\
		&=\sup_{q\in G^l}\sup_{\substack{\psi\in L^2((\g^*)^l)\\\quadnorm{\psi}=1}}\quadnorm{\sum \int d\mu_k(\xi)g_k(q) e^{i\hbar\xi\cdot}\psi(\cdot)}\\
		&=\sup_{q\in G^l}\sup_{\substack{\hat\psi\in L^2(\g^l)\\\quadnorm{\hat\psi}=1}}\quadnorm{\sum \int d\mu_k(\xi)g_k(q)\hat\psi(\cdot+\hbar\xi)},
	\end{align*}
	by using Parseval's identity twice in the last step. The lemma follows.
\end{proof}

\begin{prop}\label{prop:Rieffel0}
	\emph{(Rieffel's condition at 0)} For each $f\in\A_0^\infty$, we have
			\begin{align*}
				\lim_{\hbar\to0}\norm{\Q(f)}=\supnorm{f}.
			\end{align*}
\end{prop}
\begin{proof}
	Let $f\circ\gamma_l\in\A^\infty_0$ be arbitrary, for arbitrary $l\in\L$ and $f\in\M^l_0$. Write $f=\sum_{k=1}^K g_k\otimes \hat\mu_k$. 
		We need to prove that $\norm{\Q(f\circ\gamma_l)}$ converges to $\supnorm{f\circ\gamma_l}=\supnorm{f}$, which by \eqref{eq:Q^infty sup over Q^j} comes down to proving that $\norm{\Qm(f\circ\gamma_{lm})}$ converges to $\supnorm{f}$ \emph{uniformly} in $m$. %This makes it more tricky than the situation in \cite{stienstra19,vnuland20}.
		
		For proving $\lim_{\hbar\to0} \norm{\Q(f\circ\gamma_l)}\geq \supnorm{f}$, we can simply use the similar statement for $\Ql$. Indeed, \cite[Theorem 7.8(1)]{Stienstra} gives
		\begin{align*}
			\lim_{\hbar\to0}\norm{\Q(f\circ\gamma_l)}=\lim_{\hbar\to0}\sup_{m\in\L_{\geq l}}\norm{\Qm(f\circ\gamma_{lm})}\geq\lim_{\hbar\to0}\norm{\Ql(f)}\geq \supnorm{f}.
		\end{align*}
		
		The reverse inequality, however, is considerably more difficult. For any $\epsilon>0$, we will need to construct an $\hbar_0>0$ such that for all $|\hbar|\leq\hbar_0$ we have $\norm{\Qm(f\circ\gamma_{lm})}\leq \supnorm{f}+\epsilon$, independently of $m$.

		Let $\epsilon>0$ be arbitrary. We define
		\begin{align}\label{eq:Q}
			Q:=\sum_{k=1}^K\supnorm{g_k}\absnorm{\mu_k},	
		\end{align}
		and remark that $\norm{\Q(f\circ\gamma_l)}\leq Q$ for all $\hbar\in[-1,1]$. Pick $N\in\N$ and distinct points $x_1,\ldots,x_N\in G^l$ such that for $$r:=\sup_{y\in G^l}\inf_{j=1}^N d(y,x_j),$$
		we have $B_{r}[0_{\g^l}]\subseteq (B_{1/2}[0_\g])^l$, as well as $r<1/4$ and
		\begin{align}\label{eq:uniform continuity}
			d(x,y)<2r \Rightarrow |g_k(x)-g_k(y)|<\frac{\epsilon}{12Q\sum_k\absnorm{\mu_k}}.
		\end{align}
		We define, for all $j\in\{1,\ldots,N\}$ and $\delta\geq0$, the sets
		\begin{align*}
			V_{\delta,j}:=\{y\in G^l:~d(y,x_j)+\delta\leq d(y,x_{j'})\text{ for all } j'\neq j\}.
		\end{align*}
		We have $V_{\delta,j}\subseteq V_{0,j}\subseteq L_{x_j}(B_{1/2}[0]^l)$.
		Choose $\delta>0$ such that $\delta\leq r$ and 
		\begin{align*}
			\vol(G^l\setminus\cup_{j=1}^N V_{\delta,j})<\frac{\epsilon}{3Q^2}.
		\end{align*}
		Choose $\hbar_0>0$ such that 
		\begin{align}\label{eq:hbar0}
			\max_{\xi\in\cup_k\supp(\mu_k)}\norm{\hbar_0\xi}<\frac{\delta}{2}.
		\end{align}
		Let $\hbar\in[-1,1]$ with $|\hbar|<\hbar_0$ be arbitrary. 
		Let $n\geq l$ be arbitrary. Let $m\in\L$ be the (unique) lattice for which $l\leq m\leq n$, $m\subseteq n$ and $m\setminus\{e\}\ngeq l$ for all $e\in m$, i.e., $m$ is made from $l$ by subdivisions, and $n$ is made from $m$ by additions of edges. As $F^\add_Q$ is isometric, 
		\begin{align}\label{eq:j or k same same}
			\norm{\Qn(f\circ\gamma_{ln})}=\norm{F_Q^{nm}(\Qm(f\circ\gamma_{lm}))}=\norm{\Qm(f\circ\gamma_{lm})},
		\end{align}
		so it suffices to prove that $\norm{\Qm(f\circ\gamma_{lm})}\leq\supnorm{f}+\epsilon$.
		Define 
		\begin{align*}
			\tilde V_{\delta,j}:=(\gamma_{lm}^\conf)^{-1}(V_{\delta,j}),\qquad\tilde V:=\cup_{j=1}^N\tilde V_{\delta,j},
		\end{align*}
		as depicted in Figure \ref{fig:1a}. %$\gamma^{-1}(\cup_{p=1}^N V_{\delta,p})$ 
%		Furthermore, define
%		\begin{align*}
%			\tilde A:=G^m\setminus\tilde V=(\gamma_{lm}^\conf)^{-1}(G^l\setminus\cup_{j=1}^NV_{\delta,j}).
%		\end{align*}
		It is easily checked that $U\mapsto  (\gamma_{lm}^\conf)^{-1}(U)$ preserves volume. Hence $\vol(G^m\setminus\tilde V)<\epsilon/(3Q^2)$.
\begin{figure}[!htb]
	\centering
	\thicklines
	\begin{subfigure}{.49\textwidth}
		\centering
		\scalebox{0.7}{
		\begin{picture}(200,260)(0,0)
	%vierkant
			\put(0,0){\line(1,0){200}}
			\put(0,0){\line(0,1){200}}
			\put(200,0){\line(0,1){200}}	
			\put(0,200){\line(1,0){200}}
	
	%schuine lijnen		
			\put(0,2){\line(1,-1){2}}
			\put(0,48){\line(1,-1){48}}
			\put(0,52){\line(1,-1){52}}
			\put(0,98){\line(1,-1){98}}
			\put(0,102){\line(1,-1){102}}
			\put(0,148){\line(1,-1){148}}
			\put(0,152){\line(1,-1){152}}
			\put(0,198){\line(1,-1){198}}
			\put(1,200){\line(1,-1){198}}
			\put(48,200){\line(1,-1){152}}
			\put(52,200){\line(1,-1){148}}
			\put(98,200){\line(1,-1){102}}
			\put(102,200){\line(1,-1){98}}
			\put(148,200){\line(1,-1){52}}
			\put(152,200){\line(1,-1){48}}
			\put(198,200){\line(1,-1){2}}
			
			\put(52,60){\large $\tilde V_{\delta,3}$}
			\put(52,110){\large $\tilde V_{\delta,4}$}
			\put(120,96){\large $\tilde V_{\delta,1}$}
			\put(120,144){\large $\tilde V_{\delta,2}$}
			
	%lijn erboven incl opdeling
			\put(0,250){\line(1,0){200}}
			\put(2,254){\line(0,-1){8}}
			\put(25,250){\circle*{4}}
			\put(20,258){\large $x_1$}
			\put(10,230){\large $V_{\delta,1}$}
			\put(48,254){\line(0,-1){8}}
			\put(52,254){\line(0,-1){8}}
			\put(75,250){\circle*{4}}
			\put(70,258){\large $x_2$}
			\put(60,230){\large $V_{\delta,2}$}
			\put(98,254){\line(0,-1){8}}
			\put(102,254){\line(0,-1){8}}
			\put(125,250){\circle*{4}}
			\put(120,258){\large $x_3$}
			\put(110,230){\large $V_{\delta,3}$}
			\put(148,254){\line(0,-1){8}}
			\put(152,254){\line(0,-1){8}}
			\put(175,250){\circle*{4}}
			\put(170,258){\large $x_4$}
			\put(160,230){\large $V_{\delta,4}$}
			\put(198,254){\line(0,-1){8}}	
		\end{picture}
		}	
		\caption{The subspaces $\tilde V_{\delta,j}:=(\gamma_{lm}^\conf)^{-1}(V_{\delta,j})$.}
		\label{fig:1a}
	\end{subfigure}
	\begin{subfigure}{0.49\textwidth}
		\centering
		\scalebox{0.7}{
		\begin{picture}(200,260)(0,0)
	%vierkant
			\put(0,0){\line(1,0){200}}
			\put(0,0){\line(0,1){200}}
			\put(200,0){\line(0,1){200}}	
			\put(0,200){\line(1,0){200}}
	
	%lijn erboven incl opdeling
			\put(0,250){\line(1,0){200}}
			\put(0,254){\line(0,-1){8}}
			\put(25,250){\circle*{4}}
			\put(50,254){\line(0,-1){8}}
			\put(75,250){\circle*{4}}
			\put(100,254){\line(0,-1){8}}
			\put(125,250){\circle*{4}}
			\put(120,258){\large $x_j$}
			\put(110,230){\large $V_{0,j}$}
			\put(150,254){\line(0,-1){8}}
			\put(175,250){\circle*{4}}
			\put(200,254){\line(0,-1){8}}	
			
	%schuine lijnen: V_{0,p}
			\put(0,50){\line(1,-1){50}}
			\put(0,100){\line(1,-1){100}}
			\put(0,150){\line(1,-1){150}}
			\put(0,200){\line(1,-1){200}}
			\put(50,200){\line(1,-1){150}}
			\put(100,200){\line(1,-1){100}}
			\put(150,200){\line(1,-1){50}}
			
	%U_s
			\put(0,100){\line(1,3){12.5}}
			\put(40,60){\line(1,3){12.5}}
			\put(80,20){\line(1,3){12.5}}
			\put(126.666667,0){\line(1,3){5.833333}}
			\put(120,180){\line(1,3){6.666667}}
			\put(160,140){\line(1,3){12.5}}
			\put(20,95){\large $U_1$}
			\put(60,55){\large $U_2$}
			\put(100,15){\large $U_3$}
			\put(140,175){\large $\ddots$}
			\put(176,135){\large $U_M$}
		\end{picture}
		}
		\caption{The subspaces $U_1,\ldots,U_M$ for a fixed $j$.}
		\label{fig:1b}
	\end{subfigure}
	\caption{Dividing the configuration space $G^m\cong\T^2$ into small subspaces when $m$ has two edges (of different length) and $l$ has one.}
\end{figure}
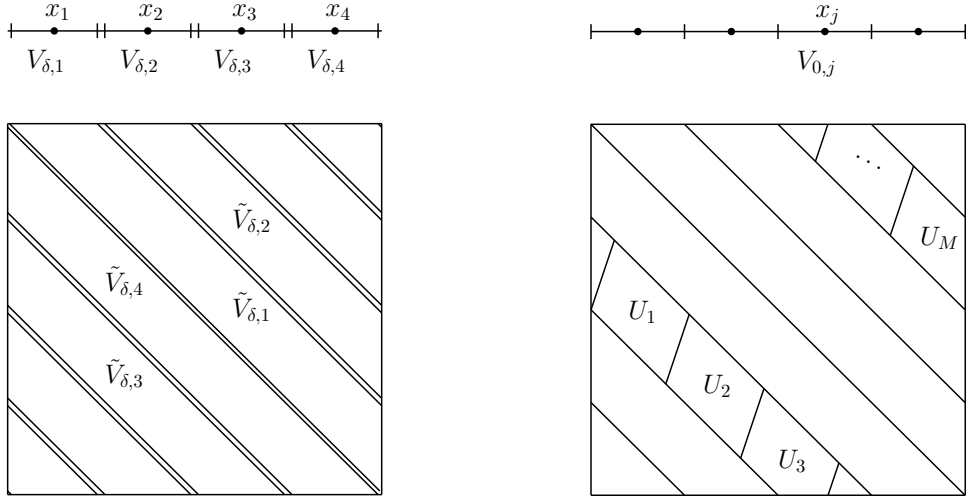
		Choose $\psi\in\H^m$ such that $\quadnorm{\psi}=1$ and
		\begin{align}\label{eq:eps over 3}
			\quadnorm{\Qm(f\circ\gamma_{lm})\psi}^2\geq \norm{\Qm(f\circ\gamma_{lm})}^2-\frac{\epsilon}{3}.
		\end{align}
		We now claim that there exists a point $q_0\in G^m$ such that 
		\begin{align}\label{eq:contr}
			\int_{G^m\setminus L_{q_0}(\tilde V)}\!dq\,|\Qm(f\circ\gamma_{lm})\psi(q)|^2\leq \vol(G^m\setminus\tilde V)Q^2<\frac{\epsilon}{3}.
		\end{align}
		Indeed, if there were no such $q_0\in G^m$, we would obtain
		\begin{align*}
			\vol(G^m\setminus\tilde V)Q^2&<\int_{G^m}dq_0\int_{G^m\setminus L_{q_0}(\tilde V)}\! dq\,|\Qm(f\circ\gamma_{lm})\psi(q)|^2\\
			&=\int_{G^m\setminus \tilde V}\!dq\int_{G^m}dq_0\,|\Qm(f\circ\gamma_{lm})\psi(q_0q)|^2\\
			&=\vol(G^m\setminus\tilde V)\quadnorm{\Qm(f\circ\gamma_{lm})\psi}^2\leq\vol(G^m\setminus\tilde V)Q^2,
		\end{align*}
		which is a contradiction. Therefore a $q_0\in G^m$ satisfying \eqref{eq:contr} does exist, and is fixed throughout the rest of the proof. Using \eqref{eq:eps over 3}, we conclude
		\begin{align}
			\norm{\Qm(f\circ\gamma_{lm})}^2-\frac{2\epsilon}{3}&\leq\quadnorm{\Qm(f\circ\gamma_{lm})\psi}^2-\frac{\epsilon}{3}\nonumber\\
			&<\sum_{j=1}^N\int_{L_{q_0}(\tilde V_{\delta,j})}dq\,|\Qm(f\circ\gamma_{lm})\psi(q)|^2.%\\
%			&=\sum_{j=1}^N\int_{L_{q_0}(\tilde V_{\delta,j})}dq\,\Big|\sum_{k=1}^Kg_k(\gamma_{lm}^\conf(q+\tfrac12 S^{ml}(\hbar\xi_k)))\psi(q+S^{ml}(\hbar\xi_k))\Big|^2.
		\label{eq:Rieffel0 phase 1}
		\end{align}
		For all $j\in\{1,\ldots,N\}$, we define 
		\begin{align*}
			\psi_j:=\psi \indicator_{L_{q_0}(\tilde V_{0,j})}.
		\end{align*}
		By $|\hbar|\leq\hbar_0$ and \eqref{eq:hbar0}, we have $\norm{\hbar\xi}<\delta/2$ for all $\xi\in\cup_k\supp(\mu_k)$. By using $\gamma_{lm}^\conf[S^{ml}(\hbar\xi)]=[\hbar\xi]$ we infer that $q\in L_{q_0}(\tilde V_{\delta,j})$ implies $q+S^{ml}(\hbar\xi)\in L_{q_0}(\tilde V_{0,j})$. Therefore, by using
		\begin{align}\label{eq:Qj}
			\Qm(f\circ\gamma_{lm})\psi(q)=\sum_{k=1}^K\int d\mu_k(\xi) g_k(\gamma_{lm}^\conf(q+\tfrac12 S^{ml}(\hbar\xi)))\psi(q+S^{ml}(\hbar\xi)),
		\end{align}
		we obtain that, for all $q\in L_{q_0}(\tilde{V}_{\delta,j})$,
		\begin{align*}
			\Qm(f\circ\gamma_{lm})\psi(q)=\Qm(f\circ\gamma_{lm})\psi_j(q).
		\end{align*}				
		Hence, \eqref{eq:Rieffel0 phase 1} becomes
		\begin{align*}
			\norm{\Qm(f\circ\gamma_{lm})}^2-\frac{2\epsilon}{3}\leq  \sum_{j=1}^N\int_{L_{q_0}(\tilde V_{\delta,j})}dq\,|\Qm(f\circ\gamma_{lm})\psi_j(q)|^2.
		\end{align*}
		By an argument similar to how we found $q_0$ (finding a contradiction if it would not exist) now using $\sum\|\psi_j\|_2^2=\|\psi\|_2^2$, we may fix a $j\in\{1,\ldots,N\}$ such that
		\begin{align}\label{eq:bound psi_m}
			\int_{L_{q_0}(\tilde V_{\delta,j})}dq\,|\Qm(f\circ\gamma_{lm})\psi_j(q)|^2\geq\quadnorm{\psi_j}^2\left(\norm{\Qm(f\circ\gamma_{lm})}^2-\frac{2\epsilon}{3}\right).
		\end{align}
		
		We fix subspaces $U_1,\ldots,U_M\subseteq\tilde V_{0,j}$ and points $y_1,\ldots,y_M\in (\gamma_{lm}^\conf)^{-1}(\{x_j\})\subseteq\tilde V_{0,j}$ such that $y_s\in U_s$ for all $s=1,\ldots,M$ and such that
		\begin{enumerate}[(a)]
			\item\label{item:a} $\bigcup_{s}U_s=\tilde V_{0,j}$ and the $U_s$ are disjoint;
			\item\label{item:b} $L_{[S^{ml}\xi]}(U_s\cap\tilde V_{\delta,j})\subseteq U_s$ for all $\xi\in B_{\delta/2}(0)\subseteq\g^l$;
			\item\label{item:c} $U_s\subseteq L_{y_s}(B^m)$ for all $s$. 
		\end{enumerate}
		An example of such sets is depicted in Figure \ref{fig:1b}.
		Define, for all $s$,
		\begin{align*}
			\psi_{j,s}:=\psi_j\indicator_{L_{q_0}(U_s)}=\psi\indicator_{L_{q_0}(U_s)}.
		\end{align*}
		By \eqref{item:a}, we have
		\begin{align*}
			&\int_{L_{q_0}(\tilde V_{\delta,j})}dq\,|\Qm(f\circ\gamma_{lm})\psi_j(q)|^2
			=\sum_{s=1}^M\int_{L_{q_0}(U_s\cap \tilde V_{\delta,j})}dq\,|\Qm(f\circ\gamma_{lm})\psi_j(q)|^2.
		\end{align*}
		Notice that, for all $\xi\in\cup_k\supp(\mu_k)$, we have $\hbar\xi\in B_{\delta/2}(0)$. Therefore, by \eqref{item:b}, we find that $q \in L_{q_0}(U_s\cap\tilde V_{\delta,j})$ implies that $q+ S^{ml}(\hbar\xi)\in L_{q_0}(U_s)$. Hence \eqref{eq:Qj} gives
		\begin{align*}
			&\int_{L_{q_0}(\tilde V_{\delta,j})}dq\,|\Qm(f\circ\gamma_{lm})\psi_j(q)|^2
			=\sum_{s=1}^M\int_{L_{q_0}(U_s\cap \tilde V_{\delta,j})}dq\,\Big|\Qm(f\circ\gamma_{lm})\psi_{j,s}(q)\Big|^2.
		\end{align*}
		Therefore, \eqref{eq:bound psi_m} gives
		\begin{align*}
			\sum_{s=1}^M\int_{L_{q_0}(U_s)}dq\,\Big|\Qm(f\circ\gamma_{lm})\psi_{j,s}(q)\Big|^2\geq\quadnorm{\psi_j}^2\left(\norm{\Qm(f\circ\gamma_{lm})}^2-\frac{2\epsilon}{3}\right).
		\end{align*}
		Again arguing by contradiction, and using that $\sum_{s=1}^M\quadnorm{\psi_{j,s}}^2=\quadnorm{\psi_j}^2$, we may fix an $s$ such that
		\begin{align}\label{eq:psi_ms ineq}
			\int_{L_{q_0}(U_s)}dq\,\Big|\Qm(f\circ\gamma_{lm})\psi_{j,s}(q)\Big|^2\geq\quadnorm{\psi_{j,s}}^2\left(\norm{\Qm(f\circ\gamma_{lm})}^2-\frac{2\epsilon}{3}\right).
		\end{align}
		Using the function $\psi_{j,s}\in L^2(G^m)$ we constructed, which is supported in $L_{q_0}(U_s)$, we can subsequently construct a function $\tilde\psi\in L^2(\g^m)$, as follows. First define $\breve{U}:=L_{q_0}(U_s)$ and $\breve{y}:=q_0+y_s\in \breve U$, so that the support of $q\mapsto \psi_{j,s}(\breve y+q)$ lies in $L_{\breve y}^{-1}(L_{q_0}(U_s))=L_{y_s}^{-1}(U_s)\subseteq B^m=[B_{1/2}(0_\g)^m]$ by \eqref{item:c} above. Define
		\begin{align*}
			\tilde{\psi}(X):=\begin{cases}
			\psi_{j,s}(\breve y+X)\quad&\text{if $X\in B_{1/2}(0_\g)^m$}\\
			0&\text{if $X\notin B_{1/2}(0_\g)^m$},
			\end{cases}
		\end{align*}
		which implies $\|\tilde\psi\|_2^2=\quadnorm{\psi_{j,s}}^2$. Using \eqref{eq:psi_ms ineq} we get
		\begin{align*}
		\left(\norm{\Qm(f\circ\gamma_{lm})}^2-\frac{2\epsilon}{3}\right)\big\|\tilde{\psi}\big\|_2^2&\leq \int_{\breve U}dq\,|\Qm(f\circ\gamma_{lm})\psi_{j,s}(q)|^2,
		\end{align*}
		in which we can use \eqref{eq:Qj} and expand the square of the absolute value of the sum over $k$. For brevity, we write $\dot g_k:=g_k(\gamma_{lm}^\conf(\breve y))$ and $g_{k,\xi}^q:=g_k(\gamma_{lm}^\conf(q+\tfrac12 S^{ml}(\hbar\xi)))-g_k(\gamma_{lm}^\conf(\breve y))$. We obtain
		\begin{align*}
		&\left(\norm{\Qm(f\circ\gamma_{lm})}^2-\frac{2\epsilon}{3}\right)\big\|\tilde{\psi}\big\|_2^2\\
		&\quad\leq \bigg|\sum_{k,k'=1}^K\int_{\breve U}dq\int d\overline{\mu_k}(\xi)\,\overline{(\dot g_k+g_{k,\xi}^q)\psi_{j,s}(q+S^{ml}(\hbar\xi))}\\
		&\qquad\quad\int d\mu_{k'}(\xi')(\dot g_{k'}+g_{k',\xi'}^q)\psi_{j,s}(q+S^{ml}(\hbar\xi'))\bigg|\\
		&\quad\leq \int_{\breve U}dq\,\bigg|\sum_{k=1}^K \int d\mu_k(\xi)\dot g_k\psi_{j,s}(q+S^{ml}(\hbar\xi))\bigg|^2
		\\&\qquad
		+\sum_{k,k'=1}^K\int d|\mu_k|(\xi)\int d|\mu_{k'}|(\xi')\Big(2|\dot g_k|+\sup_{q\in\breve U}|g_{k,\xi}^q|\Big)\sup_{q\in\breve U}|g_{k',\xi'}^q|\norm{\psi_{j,s}}^2.
		\end{align*}
		Because $\norm{\tfrac12\hbar\xi}<r$, because $U_s\subseteq\tilde V_{0,j}$ and because $d(x,x_j)\leq r$ for all $x\in V_{0,j}$ we can apply \eqref{eq:uniform continuity} to find, for all $k$ and $\xi\in\supp(\mu_k)$,
		\begin{align*}
			\sup_{q\in\breve U}|g_{k,\xi}^q|<\frac{\epsilon}{12Q\sum_k\absnorm{\mu_k}}.
		\end{align*}
		Therefore, and by Lemma \ref{lem:supnorm from operators},
		\begin{align*}
			&\left(\norm{\Qm(f\circ\gamma_{lm})}^2-\frac{2\epsilon}{3}\right)\big\|\tilde{\psi}\big\|_2^2\\
			&\quad\leq \int_{\g^m} dX\bigg|\sum_{k=1}^K \int d\mu_k(\xi)g_k(\gamma^\conf_{lm}(\breve y)) \tilde\psi(X+S^{ml}(\hbar \xi))\bigg|^2\\
			&\qquad+\sum_{k,k'=1}^K\absnorm{\mu_k}4\supnorm{g_k}\absnorm{\mu_{k'}}\sup_{\xi'\in\supp(\mu_{k'})}\sup_{q\in\breve U}|g_{k',\xi'}^q|\quadnorm{\psi_{j,s}}^2\\
			&\quad\leq \sup_{q\in G^m}\norm{\sum_{k=1}^K\int d\mu_k(\xi)g_k(\gamma^\conf_{lm}(q))L^*_{S^{ml}(\hbar\xi)}}^2\big\|\tilde\psi\big\|_2^2
			+\frac{\epsilon}{3}\big\|\tilde\psi\big\|_2^2\\
			&\quad=\left(\supnorm{f}+\frac{\epsilon}{3}\right)\big\|\tilde\psi\big\|_2^2.
		\end{align*}
		By \eqref{eq:j or k same same} we conclude that $\norm{\Qn(f\circ\gamma_{ln})}^2\leq\supnorm{f}+\epsilon$. Since $n\geq l$ was arbitrary, we have $\norm{\Q(f\circ\gamma_{l})}^2\leq\supnorm{f}+\epsilon$, which concludes the proof.
\end{proof}

\subsection{Rieffel's condition away from zero}
Now that we have continuity of $\hbar\mapsto\norm{\Q(f\circ\gamma_{l})}$ at $\hbar=0$, we are left to prove continuity at an arbitrary $\hbar_1\in[-1,1]\setminus\{0\}$. In the rest of the paper, we fix such an $\hbar_1$, as well as a function $f\in\M_0^{l}$, expanded as $f=\sum_{k=1}^{K}g_k\otimes \hat\mu_k$. 

%A key ingredient in the proof is subdividing the edges of $i_0$ into sufficiently small pieces such that for the resulting lattice $i$ the components of the vectors $S^{i i_0}(\hbar\xi_k)$ can be made arbitrarily small. Then one works with the lattice $i$ just as in the proof of the previous proposition, by letting $j\geq i$ be an arbitrary lattice obtained from $i$ via subdivisions.

The reason that Rieffel's condition away from zero holds in the infinite dimensional case, as opposed to the case on a finite lattice (see \cite[Remark 23]{vNS20} for a counterexample) is that $\norm{Q^\infty_{\hbar_1}(f\circ\gamma_l)}$ is given by a supremum over lattices $m\geq l$ as shown in \eqref{eq:Q^infty sup over Q^j}. Better yet: it is also given by a supremum over lattices $m\geq l^R$, with $l^R$ from Definition \ref{def:l_R}. If we choose $R$ large enough, the components of the $S^{l^Rl}(\xi)$'s, for $\xi\in\cup_k\supp(\mu_k)$, become arbitrarily small. We take advantage of this fact by the following construction.
%\begin{defi}\label{def:l_R}
%	For any natural number $R$, let $l^R\geq l$ be the lattice obtained by dividing every edge of $l$ into $R$ edges of equal length. 
%\end{defi}	
	For every edge $e\in l$, we choose a single edge $e'\in l^R$ that lies inside $e$. The edge $e'$ has a length $1/R$ times the length of $e$, %Formally, we choose an $e_f$ such that $i\setminus\{e_f\}\ngeq i_0$.  
	so we have $S^{l^Rl}(\xi)_{e'}=\frac{1}{R}\xi_e$.
	We then define the projection
	\begin{align*}
		\chi_{ll^R}:G^{l^R}\to G^{l},\quad \chi_{ll^R}(q)_e:=q_{e'},
	\end{align*}		
	and note that it satisfies $\chi_{ll^R}[S^{l^Rl}(\xi)]=[\tfrac1R  \xi]$.
	
	Similarly to the proof of Proposition \ref{prop:Rieffel0}, we will define subsets of $G^l$, which in volume approximate the whole of $G^l$ but are topologically better behaved than $G^l$.
\begin{defi}\label{def:U_delta en V_delta}
For all $\delta\geq0$, define
	\begin{align*}
		U_\delta:=\{\xi \in\g^{l}:~\xi_e\in(-\tfrac12+\tfrac12\delta,\tfrac12-\tfrac12\delta)^n \text{ for all $e\in l$}\}.
	\end{align*}
	In particular, $U_0$ is the open unit cube around 0. Using $U_\delta$, we define a subset $V_\delta\subseteq G^{l}$ with volume $\vol(V_\delta)=(1-\delta)^{n|l|}$, $n=\dim G$, by setting
	\begin{align*}
		V_\delta:=\{[\xi]\in G^{l}:~\xi\in U_\delta\}.
	\end{align*}
\end{defi}
Using these we will define subsets of $G^m$, for a particular class of lattices $m\geq l^R$. %We will eventually take a supremum over $m$ to obtain the operator norm as in \eqref{eq:Q^infty sup over Q^j}.

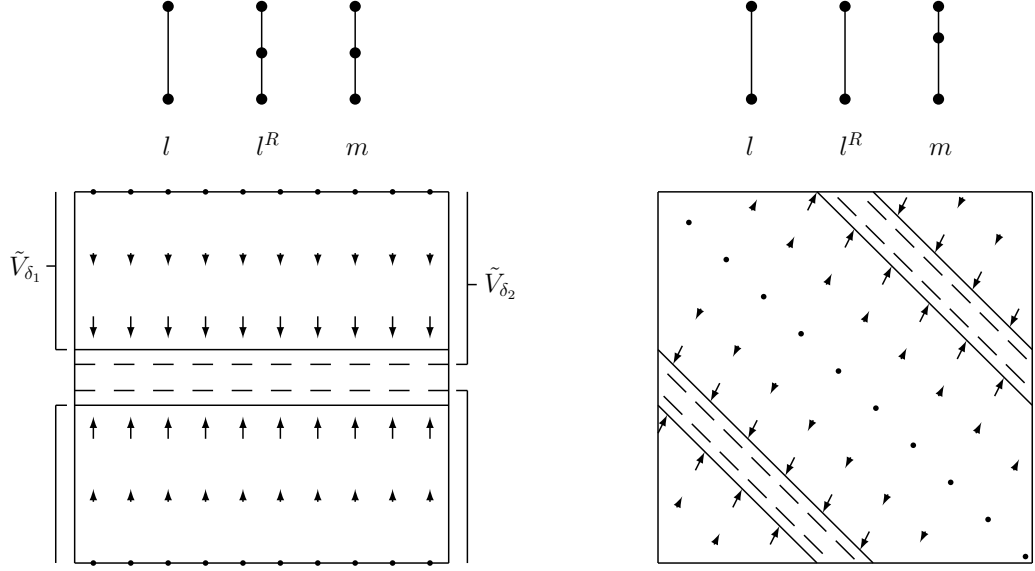
\begin{figure}[!htb]
	\centering
	\thicklines
	\begin{subfigure}{.49\textwidth}
		\centering
		\scalebox{0.7}{
		\begin{picture}(200,320)
	%roosters l, l^R, m
			\put(50,250){\line(0,1){50}}
			\put(50,250){\circle*{6}}
			\put(50,300){\circle*{6}}
			\put(100,250){\line(0,1){50}}
			\put(100,250){\circle*{6}}
			\put(100,275){\circle*{6}}
			\put(100,300){\circle*{6}}
			\put(150,250){\line(0,1){50}}
			\put(150,250){\circle*{6}}
			\put(150,275){\circle*{6}}
			\put(150,300){\circle*{6}}
			\put(47,220){\large $l$}
			\put(97,220){\large $l^R$}
			\put(145,220){\large $m$}
		
	%vierkant
			\put(0,0){\line(1,0){200}}
			\put(0,0){\line(0,1){200}}
			\put(200,0){\line(0,1){200}}	
			\put(0,200){\line(1,0){200}}

	%lijnen
			\put(0,85){\line(1,0){200}}
			\put(0,115){\line(1,0){200}}
			
	%dashed lijnen
			\multiput(0,93)(21,0){10}{\line(1,0){11}}
			\multiput(0,107)(21,0){10}{\line(1,0){11}}
	
	%accolade/curly braces
%		\put(-8,115){\rotatebox{90}{\makebox(85,5){\downbracefill}}}	
%		\put(-8,0){\rotatebox{90}{\makebox(85,5){\downbracefill}}}				
		\put(-35,155){\large $\tilde V_{\delta_1}$}
%		\put(204,107){\rotatebox{90}{\makebox(93,5){\upbracefill}}}	
%		\put(204,0){\rotatebox{90}{\makebox(93,5){\upbracefill}}}	
		\put(219,146){\large $\tilde V_{\delta_2}$}
%%%%%%	\put(-11,64){$\Bigg\lceil$}
%%%%%%%	\put(-11,40){$\Bigg|$}
%%%%%%%	\put(-12,20){~{\vrule width 1pt}}
%%%%%%	\put(-22,40){\begin{tabular}{c?{0.28mm}c}
%%%%%%~&~\\~&~\\~&~\\~&~\\~&~
%%%%%%\end{tabular}}
\put(-10,85){\line(1,0){6}}
\put(-10,85){\line(0,-1){85}}
\put(-10,115){\line(1,0){6}}
\put(-10,115){\line(0,1){85}}
\put(-10,160){\line(-1,0){5}}

\put(210,93){\line(-1,0){6}}
\put(210,93){\line(0,-1){93}}
\put(210,107){\line(-1,0){6}}
\put(210,107){\line(0,1){93}}
\put(210,150){\line(1,0){5}}	
			
	%vectoren
			%\multiput(10,0)(20,0){10}{\circle*{3}}
			%\multiput(10,200)(20,0){10}{\circle*{3}}
			\multiput(10,67)(20,0){10}{\vector(0,1){12}}
			\multiput(10,133)(20,0){10}{\vector(0,-1){12}}
			\multiput(10,33)(20,0){10}{\vector(0,1){7}}
			\multiput(10,167)(20,0){10}{\vector(0,-1){7}}
	
	%puntjes
			%\qbezier(98,0)(100,2)(102,0)
			%\put(100,1){\oval(3,3)[t]}
			\multiput(10,0)(20,0){10}{\circle*{3}}
			\multiput(10,200)(20,0){10}{\circle*{3}}
			%\multiput(9,-2)(20,0){10}{\includegraphics[width=0.3em]{bovencirkel.pdf}}
			%\multiput(8,198)(20,0){10}{\includegraphics[width=0.3em]{ondercirkel.pdf}}

		\end{picture}
		}	
		\caption{Subdivision where $R=2$ and $m=l^R$.}
		\label{fig:2a}
	\end{subfigure}
	\begin{subfigure}{0.49\textwidth}
		\centering
		\scalebox{0.7}{
		\begin{picture}(200,320)
	%roosters l, l^R, m
			\put(50,250){\line(0,1){50}}
			\put(50,250){\circle*{6}}
			\put(50,300){\circle*{6}}
			\put(100,250){\line(0,1){50}}
			\put(100,250){\circle*{6}}
			\put(100,300){\circle*{6}}
			\put(150,250){\line(0,1){50}}
			\put(150,250){\circle*{6}}
			\put(150,283){\circle*{6}}
			\put(150,300){\circle*{6}}
			\put(47,220){\large $l$}
			\put(97,220){\large $l^R$}
			\put(145,220){\large $m$}

	%vierkant
			\put(0,0){\line(1,0){200}}
			\put(0,0){\line(0,1){200}}
			\put(200,0){\line(0,1){200}}	
			\put(0,200){\line(1,0){200}}
			
	%schuine lijnen
			\put(0,85){\line(1,-1){85}}
			\put(0,115){\line(1,-1){115}}
			
			\put(85,200){\line(1,-1){115}}
			\put(115,200){\line(1,-1){85}}

	%vectoren
	
			\multiput(80,190)(20,-20){6}{\vector(1,2){5}}
			\multiput(50,190)(20,-20){8}{\vector(1,2){3}}
			\multiput(16.5,183.5)(20,-20){10}{\circle*{3}}
			\multiput(23,137)(20,-20){7}{\vector(-1,-2){3}}
			\multiput(13,117)(20,-20){6}{\vector(-1,-2){5}}
			
			\multiput(0,70)(20,-20){4}{\vector(1,2){5}}
			\multiput(10,30)(20,-20){2}{\vector(1,2){3}}	
			\multiput(163,197)(20,-20){2}{\vector(-1,-2){3}}
			\multiput(133,197)(20,-20){4}{\vector(-1,-2){5}}
			
	%dashed lijnen
			\multiput(96,197)(15,-15){7}{\line(1,-1){10}}
			%\multiput(96,199)(2,-2){52}{\circle{1}}
			\multiput(112,195)(15,-15){6}{\line(1,-1){10}}
			\multiput(5,102)(15,-15){7}{\line(1,-1){10}}
			\multiput(3,90)(15,-15){6}{\line(1,-1){10}}
%			\multiput(100,200)(10,-10){10}{\line(1,0){5}}
%			\put(100,200){\line(1,-1){10}}
		\end{picture}
		}
		\caption{Subdivision where $R=1$ and $m>l^R$.}
		\label{fig:2b}
	\end{subfigure}
	\caption{Choosing subsets $\tilde V_\delta$ ($\delta>0$) within the configuration space $G^m\cong \T^2$ such that, when $\delta_1>\delta_2$, $\tilde V_{\delta_1}\subseteq\tilde V_{\delta_2}$. The bijection $F:\tilde V_{\delta_1}\to\tilde V_{\delta_2}$ expands the subset $\tilde V_{\delta_1}$ onto $\tilde V_{\delta_2}$ along the direction of $S^{ml}\circ\uvp$, as indicated by the arrows. Here $m$ has two edges (of possibly different length) and $l$ has one.}
	\label{fig:2}
\end{figure}

\begin{lem}\label{lem:phi smooth}
	Given a  lattice $m$ obtained from $l^R$ by subdivisions (hence in particular $l\leq l^R\leq m$) the map $\varphi_{ll^Rm}:G^m\to G^{l}$ defined by
		$$\varphi_{ll^Rm}:=\chi_{ll^R}\circ\gamma^\conf_{l^Rm}:G^m\to G^l$$
	is smooth, and $U\mapsto\varphi_{ll^Rm}^{-1}(U)$ preserves volume.
\end{lem}
\begin{proof}
	By first considering the elementary steps of adding and subdividing an edge, one finds that both $\chi_{ll^R}$ and $\gamma_{l^Rm}^\conf$ are smooth and preserve volume by inverse image.
\end{proof}

	 For any $\delta\geq0$, we set
	\begin{align*}
		\tilde{V}_{\delta}:=\varphi_{ll^Rm}^{-1}(V_{\delta})\subseteq G^m.
	\end{align*}
	For $\hbar\in[-1,1]$ of the same sign as $\hbar_1$, we define a map $F:\tilde V_0\to G^m$ by
	\begin{align}\label{eq:F}
		F(q):=q+R\left(\frac{\hbar}{\hbar_1}-1\right)S^{ml}(\uvp(q)),
	\end{align}
	where $\uvp:\tilde V_0\to\g^{l}$ is defined by
	\begin{align*}
		\uvp(q):=\xi\in\g^{l}\quad\text{if}\quad \varphi_{ll^Rm}(q)=[\xi]\in G^{l}\quad\text{for}\quad \xi\in U_0.
	\end{align*}
	
\begin{lem}\label{lem:F}
	Let $\hbar\in[-1,1]$ be of the same sign as $\hbar_1$ and let $\delta_1,\delta_2\in(0,1)$ satisfy
	\begin{align}\label{eq:hbars deltas}
		\hbar(1-\delta_1)=\hbar_1(1-\delta_2).
	\end{align}
	Then $F$ restricts to a diffeomorphism $F:\tilde{V}_{\delta_1}\to\tilde V_{\delta_2}$ satisfying for all $q\in\tilde V_{\delta_1}$:
	\begin{align*}
		|\det d_qF|=(\hbar/\hbar_1)^{n|l|}.
	\end{align*}
	Moreover, when $q+tS^{ml}(\hbar_1\xi)\in\tilde V_{\delta_1}$ for all $t \in[0,1]$, we have	
	\begin{align}\label{eq:the point of F}
		F(q+S^{ml}(\hbar_1\xi))=F(q)+S^{ml}(\hbar\xi).
	\end{align}
\end{lem}
\begin{proof}
	By Lemma \ref{lem:phi smooth}, $\uvp$ is smooth, which implies that $F$ is smooth. It follows from \eqref{eq:F} that the map $\tilde V_{\delta_1}\to \text{End}(\g^m)$,  $q\mapsto d_qF$ is constant, and that
	\begin{align*}
		\varphi_{ll^Rm}(F[x])=\left[(\hbar/\hbar_1)x\right],
	\end{align*}		
	which by \eqref{eq:hbars deltas} implies that $F:\tilde V_{\delta_1}\to\tilde V_{\delta_2}$ is bijective. Therefore $F$ is a diffeomorphism and 
%	$|\det d_q F|$ is given by
%	\begin{align*}
%		\vol(\tilde V_{\delta_1})/\vol(\tilde V_{\delta_2})=\vol(V_{\delta_1})/\vol(V_{\delta_2})=(\hbar/\hbar_1)^{n|l|},
%	\end{align*}
$|\det d_q F|$ is given by
%	\begin{align*}
		$\vol(\tilde V_{\delta_1})/\vol(\tilde V_{\delta_2})=\vol(V_{\delta_1})/\vol(V_{\delta_2})=(\hbar/\hbar_1)^{n|l|},$
%	\end{align*}
	by use of Definition \ref{def:U_delta en V_delta} and Lemma \ref{lem:phi smooth}. The last statement of the lemma is a simple check.
	%	 We note that $\norm{\uvp(q)}\leq \sqrt{|l|}/2$ for all $q\in\tilde{V_0}$ and that $S^{l^R l}(\xi)=(\tfrac{1}{R}\xi,\tfrac1R\xi,\ldots,\tfrac1R\xi)$ implies
%	\begin{align*}
%		\varphi\left[R\left(\frac{\hbar}{\hbar_1}-1\right)S^{ml}(\uvp(q))\right]%&=R\left(\frac{\hbar}{\hbar_1}-1\right)\varphi[S^{ml}(\uvp(q))]\\
%		&=R\left(\frac{\hbar}{\hbar_1}-1\right)\chi[S^{l^Rl}(\uvp(q))]%\\&
%		=\left(\frac{\hbar}{\hbar_1}-1\right)\varphi(q).
%	\end{align*}
\end{proof}
In Figure \ref{fig:2}, two key examples show how $F$ maps the points of $\tilde V_{\delta_1}$ to $\tilde V_{\delta_2}$. We now have all the tools we need to establish the last part of our main result.

\begin{prop}\label{prop:Rieffel1}
	\emph{(Rieffel's condition away from 0)} For each $f\in\A_0^\infty$, and each $\hbar_1\in[-1,1]\setminus\{0\}$, we have
			\begin{align*}
				\lim_{\hbar\to\hbar_1}\norm{\Q(f)}=\norm{Q^\infty_{\hbar_1}(f)}.
			\end{align*}
\end{prop}
\begin{proof}
	Let $f\in\M_0^{l}$ for some $l\in\L$, write $f=\sum_{k=1}^K g_k\otimes \hat\mu_k$ for $g_k\in C^\infty(G^{l})$ and $\mu_k$ a finite complex measure supported in $B^{l}$, and let $\hbar_1\in[-1,1]\setminus\{0\}$. 
	By \cite[Proposition 24]{vNS20} we already have
	\begin{align*}
		\lim_{\hbar\to\hbar_1}\norm{\Q(f)}\geq\norm{Q^\infty_{\hbar_1}(f)}.
	\end{align*}
	In order to also prove 
	\begin{align*}
		\lim_{\hbar\to\hbar_1}\norm{\Q(f)}\leq\norm{Q^\infty_{\hbar_1}(f)},
	\end{align*}
	we let $\epsilon>0$ be arbitrary. By Definition \ref{def:U_delta en V_delta} and Lemma \ref{lem:phi smooth}
	we can choose $\delta\in(0,1)$ small enough such that, with $Q$ from \eqref{eq:Q},
	\begin{align}\label{eq:vol assump hbar1}
		\vol(G^m\setminus\tilde V_\delta)=\vol(G^{l}\setminus V_\delta)<\frac{\epsilon}{3Q^2}.
	\end{align}
	Choose a natural number $R\in\N$, big enough such that
	\begin{align}\label{eq:R}
		\frac{1}{R}\sqrt{|l|}<\delta,
	\end{align}
	where $|l|$ denotes the number of edges in $l$.
	For all $\xi\in\cup_k\supp(\mu_k)$, we have $\hbar_1\xi\in B^{l}$, which is an open set. We choose a number $c>0$ such that for all $\hbar\in[-1,1]$ with $|\hbar-\hbar_1|<c$ it holds that
	\begin{align}
		 1-\frac{\hbar}{\hbar_1}(1-\delta)\in(0,1);\qquad&\hbar\xi\in B^{l}\text{ for all $\xi\in\cup_k\supp(\mu_k)$;}\nonumber\\
		  1-\frac{\hbar}{\hbar_1}\left(1-\frac{\delta}{2}\right)\in(0,\delta);\qquad&\sgn(\hbar)=\sgn(\hbar_1);\nonumber\\
		  1-\frac{\hbar}{\hbar_1}\left(1-\frac{\delta}{4}\right)\in(0,1); \qquad&\supnorm{\nabla g_k}R\left|\frac{\hbar}{\hbar_1}-1\right|\sqrt{|l|}\leq \frac{\epsilon}{6Q\sum_k\absnorm{\mu_k}}.\label{eq:assumptions c}
	\end{align}
	Let $\hbar\in [-1,1]$ be arbitrary such that $|\hbar-\hbar_1|<c$. By 
	\eqref{eq:Q^infty sup over Q^j} and \eqref{eq:j or k same same}, it suffices to prove
	\begin{align*}
		\norm{Q^m_{\hbar}(f\circ\gamma_{lm})}\leq\norm{Q^m_{\hbar_1}(f\circ\gamma_{lm})}+\epsilon,
	\end{align*}
	for all lattices $m\geq l^R$ obtained from $l^R$ purely by subdivision of edges. We let $m$ be such a lattice in the following.
	We choose a $\psi\in\H^m$ such that $\quadnorm{\psi}=1$ and
	\begin{align*}
		\norm{\Qm(f\circ\gamma_{lm})}^2-\frac{\epsilon}{3}\leq \quadnorm{\Qm(f\circ\gamma_{lm})\psi}^2.
	\end{align*}
	By a proof by contradiction (as we gave several times in the proof of Proposition \ref{prop:Rieffel0}) using \eqref{eq:vol assump hbar1} we obtain a point $q_0\in G^m$ such that
	\begin{align*}
		\int_{L_{q_0}(G^m\setminus\tilde{V}_\delta)}dq|\Qm(f\circ\gamma_{lm})\psi(q)|^2\leq \frac{\epsilon}{3}.
	\end{align*}
	Therefore
	\begin{align}\label{eq:2epsilon/3}
		\norm{\Qm(f\circ\gamma_{lm})}^2-\frac{2\epsilon}{3}\leq \int_{L_{q_0}(\tilde V_\delta)}dq\left|\Qm(f\circ\gamma_{lm})\psi(q)\right|^2.
	\end{align}		
	
%	For all $\tilde\delta\geq0$, we define
%	\begin{align*}
%		\psi_{\tilde\delta}:=\psi\indicator_{L_{q_0}(\tilde V_{\tilde{\delta}})}.
%	\end{align*}

	Using \eqref{eq:assumptions c}, define $F_{q_0}:L_{q_0}(\tilde V_{\delta/4})\to G^m$ by $F_{q_0}(q):=F(q-q_0)+q_0$, so that when $q+tS^{ml}(\hbar_1\xi)\in L_{q_0}(\tilde V_{\delta/4})$ for all $t\in[0,1]$, \eqref{eq:the point of F} gives
	\begin{align}\label{eq:Fqo}
		F_{q_0}(q+S^{ml}(\hbar_1\xi))%&=F(q-q_0+S^{ml}(\hbar_1\xi))+q_0\\
%		&=F(q-q_0)+q_0+S^{ml}(\hbar\xi)\\
		&=F_{q_0}(q)+S^{ml}(\hbar\xi).
	\end{align}
	As $\tilde V_\delta\subseteq \tilde V_{\delta/4}$, we may in particular define $\tilde\psi\in\H^m= L^2(G^m)$ by
	\begin{align*}
		\tilde\psi(q):=\begin{cases}
		\sqrt{(\hbar/\hbar_1)^{n|l|}}\psi(F_{q_0}(q))\quad&\text{if }q\in L_{q_0}(\tilde V_{\delta})\\
		0\quad&\text{if }q\in G^m\setminus L_{q_0}(\tilde V_{\delta}).
		\end{cases}
	\end{align*}
	From Lemma \ref{lem:F} and the first assumption of \eqref{eq:assumptions c} we obtain that $\|\tilde\psi\|_2^2\leq\norm{\psi}_2^2=1$.

	We have $\hbar_1\xi\in B^{l}$, so $\norm{\hbar_1\xi}\leq \sqrt{|l|}/2$ for all $\xi\in\cup_k\supp(\mu_k)$. By \eqref{eq:R}, and because $\varphi_{ll^Rm}[S^{ml}(X)]=[\tfrac1R X]$, we have
	\begin{align}\label{eq:what xi's do to V's}
		\norm{\uvp[S^{ml}(\hbar_1\xi)]}\leq \delta/2.
	\end{align}
	Therefore $q+S^{ml}(\hbar_1\xi)\in \tilde V_{\delta}$ implies $q\in\tilde V_{\delta/2}$. Translating this implication with $L_{q_0}$, we obtain,
	\begin{align}
		\norm{Q^m_{\hbar_1}(f\circ\gamma_{lm})\tilde\psi}^2=\int_{G^m}&dq\left|\sum \int d\mu_k(\xi)g_k(\gamma^\conf_{lm}(q+\tfrac12 S^{ml}(\hbar_1\xi)))\tilde\psi(q+S^{ml}(\hbar_1\xi))\right|^2\nonumber\\
		=\int_{L_{q_0}(\tilde V_{\delta/2})}&dq\left|\sum \int d\mu_k(\xi)\,\dot{g}_k^q\,\tilde\psi(q+S^{ml}(\hbar_1\xi))\right|^2\label{eq:bigint hbar1}
	\end{align}
	when we define, for all $q\in L_{q_0}(\tilde V_{\delta/2})$,
	\begin{align*}
		\dot g_k^q&:=g_k(\gamma_{lm}^\conf(q+\tfrac12 S^{ml}(\hbar_1\xi)));\\ 
		\overline g_{k,\xi}^q&:=g_k(\gamma^\conf_{lm}(F_{q_0}(q+\tfrac12 S^{ml}(\hbar_1\xi))))-g_k(\gamma_{lm}^\conf(q+\tfrac12 S^{ml}(\hbar_1\xi))).
	\end{align*}
	We choose $\delta_4$ such that $F:\tilde{V}_{\delta/2}\to\tilde V_{\delta_4}$ is a bijection by Lemma \ref{lem:F}, i.e., we define $\delta_4:=1-\hbar/\hbar_1(1-\delta/2)$. By \eqref{eq:assumptions c}, we have $\delta_4\in(0,\delta)$, and therefore $\tilde V_{\delta}\subseteq\tilde V_{\delta_4}$. When we apply a change of variables $q\mapsto F_{q_0}(q)$ to \eqref{eq:2epsilon/3} we obtain, by Lemma \ref{lem:F} and \eqref{eq:Fqo},
	\begin{align}
		&\norm{\Qm(f\circ\gamma_{lm})}^2-\frac{2\epsilon}{3}\nonumber\\
		&\quad\leq \int_{L_{q_0}(\tilde V_{\delta_4})} dq\left|\sum \int d\mu_k(\xi)g_k(\gamma_{lm}^\conf(q+\tfrac12 S^{ml}(\hbar\xi)))\psi\big(q+S^{ml}(\hbar\xi)\big)\right|^2\nonumber\\
		&\quad= \left(\frac{\hbar}{\hbar_1}\right)^{n|l|}\int_{L_{q_0}(\tilde V_{\delta/2})} dq\bigg|\sum\int d\mu_k(\xi)
		 (\dot{g}_k^q+\overline{g}_{k,\xi}^q)\psi\big(F_{q_0}(q)+S^{ml}(\hbar\xi)\big)\bigg|^2\nonumber\\
		&\quad=\int_{L_{q_0}(\tilde V_{\delta/2})}dq\left|\sum\int d\mu_k(\xi) (\dot{g}_k^q+\overline{g}_{k,\xi}^q)\tilde\psi(q+S^{ml}(\hbar_1\xi))\right|^2.\label{eq:bigint hbar}
	\end{align}
	%Note $(1-\delta)^{|i_0|n}=\vol(\tilde V_\delta)$.
	The only difference between \eqref{eq:bigint hbar1} and \eqref{eq:bigint hbar} is now the appearance of $\overline g^q_{k,\xi}$ in the latter expression.
	For all $q\in L_{q_0}(\tilde V_{\delta/2})$, $\xi\in\cup_k\supp(\mu_k)$, and $t\in[0,1]$, we have $q+\tfrac t2 S^{ml}(\hbar_1\xi)\in L_{q_0}(\tilde V_{\delta/4})$ by \eqref{eq:what xi's do to V's}.
	By Lemma \ref{lem:F} and \eqref{eq:assumptions c}, we obtain
	\begin{align*}
%		\sup_{q\in L_{q_0}(\tilde V_{\delta/2})}
		|\overline{g}^q_{k,\xi}|&\leq\supnorm{\nabla g_k} d\left(\gamma_{lm}^\conf(F_{q_0}(q+\tfrac12 S^{ml}(\hbar_1\xi))),\gamma_{lm}^\conf(q+\tfrac12 S^{ml}(\hbar_1\xi))\right)\\
		&\leq \supnorm{\nabla g_k}\norm{R\left(\frac{\hbar}{\hbar_1}-1\right)\uvp(q-q_0+\tfrac12 S^{ml}(\hbar_1\xi)}\\
		&\leq \supnorm{\nabla g_k}R\left|\frac{\hbar}{\hbar_1}-1\right|\frac{\sqrt{|l|}}{2}
		\leq\frac{\epsilon}{12Q\sum_k\absnorm{\mu_k}},
	\end{align*}
	for all $q,k,$ and $\xi$. Expanding the square of the absolute value in \eqref{eq:bigint hbar}, and using \eqref{eq:bigint hbar1},
		\begin{align*}
		&\norm{\Qm(f\circ\gamma_{lm})}^2-\frac{2\epsilon}{3}\\
		&\quad\leq \int_{L_{q_0}(\tilde V_{\delta/2})}dq\,\left|\sum_{k=1}^K \int d\mu_k(\xi)\dot\, g^q_k\,\tilde\psi(q+S^{ml}(\hbar_1\xi))\right|^2
		\\&\qquad\quad
		+\sum_{k,k'=1}^K\int d|\mu_k|(\xi)\Big(2\sup_q |\dot g^q_{k}|+\sup_{q}|\overline g_{k,\xi}^q|\Big)\int d|\mu_{k'}|(\xi')\sup_{q}|\overline g_{k',\xi'}^q|\big\|\tilde\psi\big\|_2^2\\
		&\quad\leq \quadnorm{Q^m_{\hbar_1}(f\circ\gamma_{lm})\tilde\psi}^2
		+4Q\sum_{k=1}^K\absnorm{\mu_k}\sup_{\xi}\sup_{q}|\overline g_{k,\xi}^q|\\
		&\quad\leq \norm{Q^m_{\hbar_1}(f\circ\gamma_{lm})}^2
		+\frac{\epsilon}{3}.
		\end{align*}
		Therefore $\norm{\Qm(f\circ\gamma_{lm})}^2\leq \norm{Q^m_{\hbar_1}(f\circ\gamma_{lm})}^2+\epsilon$ for all lattices $m\geq l^R\geq l$, which is what we needed to prove.
\end{proof}
We conclude that $\Q:\A^\infty_0\to\A^\infty_\hbar$ is a strict deformation quantization:

\begin{proof}[Proof of Theorem \ref{thm:sdq infinity}]
	Combine Propositions \ref{prop:algebra}, \ref{prop:star-preserving}, \ref{prop:injective}, \ref{prop:von Neumann}, \ref{prop:Dirac}, \ref{prop:Rieffel0}, and \ref{prop:Rieffel1}.
\end{proof}

\section{Outlook}\label{sct:outlook}
We constructed field C*-algebras for classical and quantum abelian gauge theories, noted their advantageous properties, and connected them via strict deformation quantization. 
%These results give hope that, incremented by the methods of this paper, the program of C*-algebraic Hamiltonian lattice gauge theory will at one point be able to non-perturbatively construct a quantum field theory for QED, and possibly even quantum Yang--Mills. Clearly, a lot still needs to be done. The steps outlined below give an indication of what needs to be done for this program to succeed, and how our methods can be of help.
These results already show that our new method, guided by quantization and using operator systems at the finite level, can be a valuable contribution to the program of C*-algebraic lattice gauge theory. We will now discuss how our method can be exploited even further, when taking three important next steps.

\paragraph{Reduction}
We have constructed field algebras for classical and quantum abelian gauge theories. The next step would be to define \emph{observable} algebras, by reduction with respect to the gauge group that is attached to each vertex of the lattice. The embedding maps from this paper respect gauge transformations, and can therefore be used to define gauge transformations on the continuum algebras. The way the classical and quantum gauge actions relate still needs to be worked out, which makes it hard to say whether any choice of observable algebra is correct at this point. The C*-algebraic quantization route looks very promising, as there one can use the fact that a restriction of a strict deformation quantization map to a *-subalgebra, whose image is multiplicatively closed, is also a strict deformation quantization.

\paragraph{Full time evolution}
The combined results of \cite{vNS20} and this paper give hope that C*-algebras like $A_0^\infty$ and $A_\hbar^\infty$ are invariant under the full abelian Yang--Mills time evolution in 3+1 dimensions.
The Hamiltonians of Theorems \ref{thm:classical time evolution infty} and \ref{thm:quantum time evolution infty} give rise to the classical and quantum time evolutions corresponding to the electric part of the Kogut--Susskind Hamiltonian, with a rescaling of each term by a factor proportional to the square of the length scale of the cubic lattice. See \cite{KS} and \cite[page 34]{Stienstra}. This rescaling can be countered in many ways; perhaps the most natural way is to reinterpret the configuration space of an edge from a singular parallel transport to an average of parallel transports, effectively assigning each edge not only a length but a width and breadth as well. Accordingly, the embedding map for adding an edge should be altered in the same way that we altered the embedding map for subdivision in this paper. The above reinterpretation might also make it easier to obtain invariance under full time evolution, by ensuring that the finite approximations (for which we can apply \cite{vNS20}) converge to the continuum in a stronger sense.

%\subsection{Extension to a larger algebras: the resolvent algebra}
%Another important step is to extend the quantization map to a larger algebra, in order to tackle full time evolution. Good candidates for the C*-algebras on a finite lattice are the commutative and noncommutative resolvent algebra on the cylinder, defined in \cite{vNS20}. Besides considering functions of the form $e^{i\xi\cdot}$, one then also considers functions $g\circ P$ for a finite projection $P$ and a rapidly decaying $g$.
%As discussed in \cite[Section 3, Section 5]{vNS20}, these larger classical and quantum C*-algebras on a finite lattice are conserved under a large class of time evolutions. An interesting question is whether this result extends to infinite dimensions, i.e., whether the direct limit algebra is conserved under Yang--Mills time evolution. To construct the limit algebras, a first step would be to assume that $\supp \hat{g}\subseteq B^l$ to obtain an operator system on the lattice $l$. 

\paragraph{Generalization to non-abelian groups}
Many of the definitions in this paper were inspired by the more general case where the gauge group $G=\T^n$ is replaced with any compact Lie group $G$. The quantization map generalizes by correctly incorporating the exponential map in \eqref{eq:Ql}. In particular,
\begin{align*}
	\Ql(g \otimes e^{i\xi\cdot}):=M_{g\circ L_{\exp(\hbar\xi/2)}}L^*_{\exp(\hbar\xi)}\in\mB(L^2(G^l)).
\end{align*}

The embedding maps for adding an edge are unaltered in the non-abelian case. The embedding maps for subdivision generalize as well, giving in particular,
\begin{align*}
	F_C^\sub(f)(q_1,q_2,p_1,p_2)&:=f(\mu(q_1,q_2),p_1+p_2),\\
	F_Q^\sub(M_gL^*_{\exp(X)})&:=M_{g\circ \mu}L^*_{\exp\left(\frac{d_1}{d}X,\frac{d_2}{d}X\right)},
\end{align*}
where $\mu:G\times G\to G$ is the group multiplication. These formulas suggest that it is possible to define a natural quantization map in the continuum limit.
It would be interesting to see whether the properties discussed in this paper still hold and, in particular, whether this map is a strict deformation quantization.

\end{document}